\documentclass[11pt,a4paper]{article}
\usepackage[T1]{fontenc}
\usepackage{lmodern}
\usepackage[margin=1in]{geometry}
\usepackage{amsmath,amssymb,amsthm,mathtools}
\usepackage{booktabs,array}
\usepackage{placeins}
\usepackage{needspace}
\usepackage{microtype}
\usepackage[hidelinks]{hyperref}
\hypersetup{pdftitle={Robust self-testing of nonmaximal entanglement from a reduced inner product game}}
\numberwithin{equation}{section}
\newtheorem{theorem}{Theorem}[section]

\newtheorem{lemma}[theorem]{Lemma}
\newtheorem{corollary}[theorem]{Corollary}
\theoremstyle{definition}

\theoremstyle{remark}

\DeclareMathOperator{\Tr}{Tr}
\DeclareMathOperator{\diag}{diag}
\newcommand{\C}{\mathbb C}
\newcommand{\R}{\mathbb R}

\newcommand{\ket}[1]{\lvert #1\rangle}
\newcommand{\bra}[1]{\langle #1\rvert}
\newcommand{\norm}[1]{\lVert #1\rVert}
\newcommand{\ip}[2]{\langle #1,#2\rangle}
\newcommand{\psistar}{\psi_\star}
\allowdisplaybreaks[2]
\title{Robust self-testing of nonmaximal entanglement\\from a reduced inner product game}
\author{
  Ranyiliu Chen\\
  {\small Quantum Science Center of Guangdong--Hong Kong--Macao Greater Bay Area}\\
  {\small\href{mailto:chenranyiliu@quantumsc.cn}{\texttt{chenranyiliu@quantumsc.cn}}}
}
\date{}
\begin{document}
\maketitle
\begin{abstract}
We reduce Lalonde's pseudo-telepathy inner product game from six to five dimensions, keeping its four Alice questions and three Bob questions and reducing each answer alphabet to five. We prove that the reduced game self-tests its nonmaximally entangled state and all 35 measurement effects. The self-test is also robust: for a strategy winning with probability at least $1-\varepsilon$ we obtain $d_{\mathrm{ext}}\le\min\{1,2^{21}\sqrt{\varepsilon}\}$, where $d_{\mathrm{ext}}=\sqrt{1-\Xi}$ and $\Xi$ is the optimal squared fidelity with the reference state under local extraction channels. A compactness argument then gives simultaneous qualitative robustness for all measurement actions on the state under common local isometries. As a consequence of the self-test, no maximally entangled state of any finite dimension can win the game perfectly.
\end{abstract}

\section*{\normalsize{Statement on the Use of Artificial Intelligence}}

The author(s) formulated the research question. Several frontier LLMs was used under human direction to generate exploratory proofs, prepare and revise the manuscript. The author(s) then refined the logical flow of the proof, guided the revision, and verified the statements. The author(s) take full responsibility for the content and correctness of this article.

\section{Introduction}
\label{sec:introduction}

In a bipartite nonlocal game, two players receive questions and return answers without communicating. The game exhibits \emph{pseudo-telepathy} if an entangled strategy wins with certainty but no classical strategy does~\cite{BrassardBroadbentTapp}. For such games with finite-dimensional perfect strategies, a basic question was whether a maximally entangled state--one whose nonzero Schmidt coefficients are all equal--always suffices for perfect play.

For general Bell inequalities, maximally entangled states need not attain the largest possible violations, as shown by Junge and Palazuelos and by Vidick and Wehner~\cite{JungePalazuelos,VidickWehner}. Whether they suffice for perfect winning proved harder to settle. Man\v{c}inska established structural conditions under which they do~\cite{Mancinska}. More recently, Renner et al. constructed nonmaximally entangled states supporting pseudo-telepathy in every local dimension at least three~\cite{RennerEtAl}. A game with such a strategy may also admit perfect play using maximal entanglement.

Lalonde resolved the general sufficiency question with the inner product game $G_{4,3,6,6}$, which has four questions for Alice, three for Bob, and six answers for each player~\cite{Lalonde}. His six-dimensional strategy wins perfectly, whereas a computer-assisted algebraic certificate rules out perfect play with a maximally entangled state of any finite dimension. Numerical evidence led him to conjecture that every perfect strategy contains the specific nonmaximally entangled state~\cite[Sections 4.1 and 5]{Lalonde}.

This stronger conclusion is a \emph{self-test}: it determines the shared state up to local isometries and an additional shared system. These freedoms reflect changes of local coordinates and extra degrees of freedom invisible to the observed probabilities~\cite{MayersYao,SupicBowles}. A strategy self-test also identifies the measurement actions under the same isometries. Robust self-testing extends this conclusion to nearly optimal play with a vanishing error~\cite{McKagueYangScarani}.

We reduce Lalonde's game to $G_{4,3,5,5}$ and prove a self-test of the full reduced strategy (Theorem~\ref{thm:strategy-exact}). The reduction removes one of the two low-amplitude coordinates and combines the measurement outcomes that coincide after compression. Specifically, the reduced game has classical value $11/12$, quantum value one. The reference non-maximally entangled state is
\begin{equation}
 \ket{\psistar}=\frac{2\ket{00}+2\ket{11}+2\ket{22}+2\ket{33}+\ket{44}}{\sqrt{17}}.
 \label{eq:target}
\end{equation}
It follows from self-testing that every perfect strategy produces the entire ideal correlation table (Corollary~\ref{cor:perfect-correlation}), and that no maximally entangled state of any finite dimension admits perfect play (Corollary~\ref{cor:nonmaximal}).

The full strategy self-test is also robust (Theorem~\ref{thm:strategy-robust}). The state estimate is quantitative, while robustness of the measurement actions follows from compactness. Specifically, for state extraction we obtain
\[
 d_{\mathrm{ext}}(\rho,\psistar)\le\min\{1,2^{21}\sqrt\varepsilon\}
 \qquad\text{whenever}\qquad
 \omega(G_{4,3,5,5};\rho,M,N)\ge1-\varepsilon
\]
(Theorem~\ref{thm:state-robust}), where $d_{\mathrm{ext}}$ is the square root of one minus the best squared fidelity achievable by local quantum channels.

These results connect two lines of work on self-testing. For nonmaximal entanglement, Bamps and Pironio proved robust self-testing of entangled two-qubit pure states using a family of Bell inequalities~\cite{BampsPironio}, and Coladangelo, Goh, and Scarani constructed correlations self-testing every finite-dimensional bipartite entangled pure state~\cite{ColadangeloGohScarani}. Cui et al. obtained nonmaximally entangled self-tests from games whose optimal winning probability is less than one~\cite{CuiEtAl}. For pseudo-telepathy, the magic square game self-tests two maximally entangled qubit pairs~\cite{WuEtAl}, and Coladangelo and Stark developed algebraic criteria for robust self-testing in games defined by systems of linear equations~\cite{ColadangeloStark}. Here, the zero-probability constraints of a single fixed game determine both unequal Schmidt coefficients and every labeled measurement effect.

The key step in our proof is the amplitude ratio. Bob's first measurement separates five orthogonal subspaces, and the forbidden answer pairs allow Alice to coherently record the same five labels. Relations among four components identify them, through local unitaries, with a common residual vector. Binary comparisons obtained by grouping outcomes of the prescribed measurements then link the remaining component to these four. Their equal-outcome and opposite-outcome conditions force it to have half the common amplitude. They also fix the off-diagonal measurement blocks, so pair sums and differences recover every individual effect.

We prove the explicit state bound by estimating these relations directly on the state vector with bounded operators. This keeps the estimates uniform even when individual outcomes have small probability. For measurement robustness, we extend the exact identification to commuting-operator representations, where Alice's and Bob's operators commute on a common Hilbert space. Compactness then supplies a uniform modulus for the fixed extraction maps. The connection between these representations and self-testing is developed in~\cite{PaddockEtAl,Zhao}. The results apply to mixed states and positive operator-valued measurements (POVMs). For measurement robustness, this extension follows from the assumption-lifting theorem of~\cite[Theorem 4.1(a)]{BaptistaEtAl}, since the reference state has full Schmidt rank.

Section~\ref{sec:setting} defines the game and states the results. Section~\ref{sec:robust-state} constructs the extraction maps and proves the explicit state bound. Section~\ref{sec:measurements} identifies the measurements and explains their qualitative robustness. Appendix~\ref{app:game} gives the game data, its reduction from Lalonde's game, and the support relations. Appendix~\ref{app:robust-transport} contains the quantitative state estimates. Appendix~\ref{sec:proof} proves the exact self-test, and Appendix~\ref{app:robust-measurements} establishes measurement robustness using the fixed extraction circuit.

\section{The game and the main results}
\label{sec:setting}

We use a five-dimensional reduction of Lalonde's inner product game~\cite{Lalonde}. All indices start at zero. Alice's questions are $x\in\{0,1,2,3\}$, Bob's are $y\in\{0,1,2\}$, and their answers are $a,b\in\{0,\ldots,4\}$. Each of the twelve question pairs has probability $1/12$. Let $D=\diag(2,2,2,2,1)$, and let $u_{xa},v_{yb}\in\R^5$ be the normalized vectors in Tables~\ref{tab:alice} and~\ref{tab:bob} of Appendix~\ref{app:game}. Following Lalonde's inner product prescription, the winning predicate is
\begin{equation}
 \mathsf w(x,y,a,b)=
 \begin{cases}
  1,&u_{xa}^{\mathsf T}Dv_{yb}\ne0,\\
  0,&u_{xa}^{\mathsf T}Dv_{yb}=0.
 \end{cases}
 \label{eq:game}
\end{equation}
We abbreviate this game to $G_5$. Appendix~\ref{app:game} derives the bases by compressing the original six-dimensional measurements, deleting zero effects, and combining coincident outcomes. It also proves that~\eqref{eq:game} accepts a new answer pair exactly when at least one pair of its old representatives was accepted.

A quantum strategy consists of finite-dimensional complex Hilbert spaces $\mathcal H_A,\mathcal H_B$, a density operator $\rho$ on their tensor product, and local POVMs $\{M_{a|x}\}_{a=0}^4$, $\{N_{b|y}\}_{b=0}^4$. Each measurement consists of positive operators summing to the identity. Its conditional probabilities and winning probability are
\begin{align}
 p(a,b|x,y)&=\Tr\bigl[\rho(M_{a|x}\otimes N_{b|y})\bigr],\nonumber\\
 \omega(G_5;\rho,M,N)&=\frac1{12}\sum_{x,y,a,b}\mathsf w(x,y,a,b)p(a,b|x,y).
 \label{eq:score}
\end{align}
Measuring $\ket{\psistar}$ in the real orthonormal bases $u_x,v_y$ gives
\begin{equation}
 p_\star(a,b|x,y)=\frac{|u_{xa}^{\mathsf T}Dv_{yb}|^2}{17},
 \label{eq:ideal-probability}
\end{equation}
and hence wins perfectly, as in the original inner product game strategy of~\cite[Proposition 3.1]{Lalonde}. The ideal measurement effects 
$A^\star_{xa}=u_{xa}u_{xa}^{\dagger},~B^\star_{yb}=v_{yb}v_{yb}^{\dagger}$ are real rank-one projections. We suppress identity factors when their action is clear. The local reduced-state support is the range of the corresponding reduced density operator.

To quantify state extraction, define
\begin{equation}
\begin{split}
 \Xi(\rho,\psistar)
 &=\max_{\Lambda_A,\Lambda_B}
    \bra{\psistar}(\Lambda_A\otimes\Lambda_B)(\rho)\ket{\psistar},\\
 d_{\mathrm{ext}}(\rho,\psistar)&=\sqrt{1-\Xi(\rho,\psistar)}.
\end{split}
\label{eq:extraction-distance}
\end{equation}
The maximum ranges over local quantum channels from the two device spaces to $\C^5$ and exists by compactness for fixed finite input dimensions. We use the squared-fidelity convention, so $\Xi$ is the optimal squared fidelity with the pure target. The channels may introduce and discard local auxiliary systems.

The first result identifies a quantitative state extraction.

\begin{theorem}[Robustness of the state]
\label{thm:state-robust}
For every finite-dimensional quantum strategy for $G_5$ and $0\le\varepsilon\le1$, if $\omega(G_5;\rho,M,N)\ge1-\varepsilon$, then $d_{\mathrm{ext}}(\rho,\psistar)
 \le\min\{1,2^{21}\sqrt\varepsilon\}.$
\end{theorem}

We remark that the constant $2^{21}$ is not optimized.

At perfect score, we identify the state and every labeled measurement effect.

\begin{theorem}[Exact self-test]
\label{thm:strategy-exact}
Every perfect finite-dimensional strategy for $G_5$ admits local isometries
\[
 W_A:\mathcal H_A\longrightarrow\C^5\otimes\mathcal K_A,
 \qquad
 W_B:\mathcal H_B\longrightarrow\C^5\otimes\mathcal K_B.
\]
For every purification $\Omega\in\mathcal H_A\otimes\mathcal H_B\otimes\mathcal H_E$ of $\rho$, there is a normalized residual vector $\chi\in\mathcal K_A\otimes\mathcal K_B\otimes\mathcal H_E$ such that the following identities hold for all $x,a,y,b$, with $\mathcal W=W_A\otimes W_B\otimes I_E$ and $\Phi=\ket{\psistar}\otimes\chi$:
\begin{equation}
\begin{aligned}
 \mathcal W\Omega&=\Phi,\\
 \mathcal W(M_{a|x}\otimes I)\Omega&=A^\star_{xa}\Phi,\\
 \mathcal W(I\otimes N_{b|y})\Omega&=B^\star_{yb}\Phi,\\
 \mathcal W(M_{a|x}\otimes N_{b|y})\Omega&=A^\star_{xa}B^\star_{yb}\Phi
\end{aligned}
\label{eq:strategy-exact-actions}
\end{equation}
\end{theorem}

The full strategy self-test also admits a qualitative robustness estimate.

\begin{theorem}[Qualitative robustness of the full strategy]
\label{thm:strategy-robust}
There is a function $r:[0,1]\to[0,2]$, independent of the local dimensions, such that $r(0)=0$ and $r(\varepsilon)\to0$ as $\varepsilon\downarrow0$, with the following property. Every finite-dimensional strategy satisfying $\omega(G_5;\rho,M,N)\ge1-\varepsilon$ admits local isometries $W_A,W_B$ as in Theorem~\ref{thm:strategy-exact}. For every purification $\Omega$ of $\rho$, there is a normalized residual vector $\chi$ such that, with $\mathcal W=W_A\otimes W_B\otimes I_E$ and $\Phi=\ket{\psistar}\otimes\chi$,
\begin{equation}
\begin{aligned}
 \norm{\mathcal W\Omega-\Phi}&\le r(\varepsilon),\\
 \norm{\mathcal W(M_{a|x}\otimes I)\Omega-A^\star_{xa}\Phi}&\le r(\varepsilon),\\
 \norm{\mathcal W(I\otimes N_{b|y})\Omega-B^\star_{yb}\Phi}&\le r(\varepsilon),\\
 \norm{\mathcal W(M_{a|x}\otimes N_{b|y})\Omega-A^\star_{xa}B^\star_{yb}\Phi}&\le r(\varepsilon)
\end{aligned}
\label{eq:strategy-robust-actions}
\end{equation}
simultaneously for all labels.
\end{theorem}

The conclusion concerns arbitrary mixed states and POVM effects, and the purification environment is untouched~\cite{BaptistaEtAl}. The proof of Theorem~\ref{thm:strategy-robust} gives a dimension-independent modulus without an explicit rate for the measurement actions.

Exact self-testing immediately implies the uniqueness of perfect correlation and the necessity of nonmaximal entanglement.

\begin{corollary}[Unique perfect correlation]
\label{cor:perfect-correlation}
Every perfect finite-dimensional strategy for $G_5$ realizes $p(a,b|x,y)=p_\star(a,b|x,y)$.
\end{corollary}
\begin{proof}
Take the inner product of the last identity in~\eqref{eq:strategy-exact-actions} with $\mathcal W\Omega=\Phi$.
\end{proof}

\begin{corollary}[Pseudo-telepathy and the necessity of nonmaximal entanglement]
\label{cor:nonmaximal}
The game $G_5$ has classical value $11/12$ and quantum value one. No maximally entangled state of any finite dimension admits a perfect strategy for $G_5$.
\end{corollary}
\begin{proof}
The reference strategy proves the quantum claim. A deterministic classical perfect strategy would be a one-dimensional perfect quantum strategy, which cannot contain the Schmidt-rank-five factor in Theorem~\ref{thm:strategy-exact}. Shared randomness cannot improve on the best deterministic strategy. Because there are twelve equiprobable question pairs, the classical value is at most $11/12$. Alice's answers $(0,1,1,2)$ and Bob's answers $(2,4,4)$ attain this bound; substitution in~\eqref{eq:game} shows that they lose only on $(x,y)=(0,0)$.

Suppose a pure maximally entangled input admitted perfect play. Its extracted state would be $\ket{\psistar}\otimes\chi$ for a pure residual vector $\chi$. Local isometries preserve the nonzero Schmidt coefficients. The input has a flat nonzero Schmidt spectrum, whereas every nonzero squared Schmidt coefficient of $\chi$ contributes output coefficients in the ratio $4:1$, by~\eqref{eq:target}. This is a contradiction.
\end{proof}

Section~\ref{sec:robust-state} proves the explicit state bound directly on unnormalized branch vectors, without assuming a lower bound on any outcome probability. Section~\ref{sec:measurements} explains the exact measurement identification and its qualitative robustness.

\section{Local extraction and the explicit state bound}
\label{sec:robust-state}

The target state has four equal Schmidt amplitudes and a fifth that is half as large. We prove Theorem~\ref{thm:state-robust} by recovering this structure from the game's forbidden outcomes. The players first record five matching labels in local registers. They then align the corresponding state components by local unitaries. We explain these two steps before combining their errors; the explicit operators and their estimates are given in Appendix~\ref{app:robust-transport}.

First suppose that the measurements are projective. Let $\Omega$ be a normalized purification of the shared state, and write the measurement projections as $A_{xa}$ and $B_{yb}$. A winning probability at least $1-\varepsilon$ means that
\[
 \sum_{\mathsf w(x,y,a,b)=0}\norm{A_{xa}B_{yb}\Omega}^{2}
 \le12\varepsilon=\delta^2,
 \qquad \delta=\sqrt{12\varepsilon}.
\]
Thus every collection of forbidden outcomes gives a vector estimate of order $\delta$. All the comparisons below use unnormalized components of $\Omega$; small outcome probabilities cause no difficulty.

Bob's measurement at $y=0$ supplies the five labels directly: write $Q_i=B_{0i}$. At perfect score, the coarse agreements in Table~\ref{tab:coarse} show how Alice recovers the same label through two successive coarse projective measurements. The first groups the projectors of question $x=0$ according to $\{0,1\}$, $\{2,3\}$, and $\{4\}$. These same three sets give the corresponding possibilities for Bob's label. In either of the first two groups, she then uses question $x=1$, grouping its answers as $\{1,2\}$, $\{3,4\}$, and $\{0\}$. These groups distinguish Bob's label sets $\{0,2\}$, $\{1,3\}$, and $\{4\}$. The two coarse labels together specify one of $0,1,2,3$; the remaining consistent branch gives label $4$.

Alice implements this sequence coherently, storing its outcomes in auxiliary registers. The two inconsistent sequences have zero weight at perfect score, but are retained at nonzero error so that the map is an isometry. Each player then applies a local unitary controlled by the recorded label. Denote these unitaries by $T_i^A,T_i^B$, with $T_0^A=T_0^B=I$, and the resulting isometries by $\mathcal V_A,\mathcal V_B$. Appendix~\ref{app:sequential-labels} gives the construction, with the resulting maps in~\eqref{eq:rs:isometries}. They are defined on the full device spaces and are isometries at every score.

Let $\xi=(\mathcal V_A\otimes\mathcal V_B\otimes I_E)\Omega$ be the output, with the two label registers placed first, and let
\[
 \eta_i=T_i^AT_i^BQ_i\Omega,
 \qquad \eta_0=Q_0\Omega.
\]
The coarse agreements ensure that Alice and Bob record nearly the same label. Lemma~\ref{lem:rs:labels} gives the precise estimate
\begin{equation}
 \left\|\xi-\sum_{i=0}^{4}\ket{ii}\otimes\eta_i\otimes\ket{\mathrm g}\right\|
 \le2\delta.
 \label{eq:rs:label-error}
\end{equation}
Here $\ket{\mathrm g}$ is the fixed auxiliary flag for consistent sequences. The bound includes both mismatched labels and the two inconsistent branches. It remains to compare the five residual vectors $\eta_i$.

For the first four branches, the forbidden outcomes identify local reflections that exchange labels $0\leftrightarrow1$, $2\leftrightarrow3$, and $0\leftrightarrow2$, $1\leftrightarrow3$. Their products connect all four labels. Using these reflections as the controlled unitaries therefore carries the four components to a common residual space. At perfect score they coincide; Appendix~\ref{app:robust-transport} proves the quantitative version
\begin{equation}
 \norm{\eta_i-\eta_0}\le10\,000\delta
 \qquad(1\le i\le3).
 \label{eq:rs:high-branch-bounds}
\end{equation}

The fifth branch is fixed by the binary tests at $x=2,3$ and $y=1,2$, using answers $2,3$ on each side. These tests couple the fifth component to one of the first four. To see why they impose the factor $1/2$, first consider perfect play. After the already established identifications of the first four components and suitable local changes of coordinates, the four conditional vectors are
\[
 -\tfrac12\eta_0+Q_4\Omega,\qquad
 -\tfrac12\eta_0+Q_4\Omega,\qquad
 +\tfrac12\eta_0+Q_4\Omega,\qquad
 -\tfrac12\eta_0+Q_4\Omega
\]
in the order $(x,y)=(2,1),(3,1),(2,2),(3,2)$. The coefficients $\pm1/2$ come from the high parts of Bob's two pair projections. These are already determined by the relations among the first four components: their coordinate vectors have entries $\pm1/2$, as shown in~\eqref{eq:rt:high-pair-projections}. Three tests require agreement of the binary outcomes, while the test at $(x,y)=(2,2)$ requires disagreement. Together, these conditions imply that the product of the two local binary reflections carries the fifth component to minus one half of the common component, by Lemma~\ref{lem:robust-two-subspaces}. Absorbing the minus sign into Bob's fifth controlled unitary gives $\eta_4=\eta_0/2$.

The same argument is stable when the score is below one. Appendix~\ref{app:robust-transport} bounds the errors in the coordinate changes and in the four binary tests, then applies the quantitative form of this product identity. It obtains
\begin{equation}
 \norm{\eta_4-\tfrac12\eta_0}\le400\,200\delta.
 \label{eq:rs:branch-bounds}
\end{equation}

We can now combine the estimates. By~\eqref{eq:rs:label-error}--\eqref{eq:rs:branch-bounds},
\begin{equation}
\begin{split}
 \left\|\xi-
 \left(\sum_{i=0}^{3}\ket{ii}+\tfrac12\ket{44}\right)
 \otimes\eta_0\otimes\ket{\mathrm g}\right\|
 &\le(2+3\cdot10\,000+400\,200)\delta\\
 &=430\,202\delta.
\end{split}
 \label{eq:rs:total-error}
\end{equation}
The comparison vector has the target state as its first tensor factor. Projecting $\xi$ onto the orthogonal complement of that target therefore gives
\begin{equation}
 \left\|\bigl[(I-\ket{\psistar}\!\bra{\psistar})\otimes I\bigr]\xi\right\|^2
 \le12(430\,202)^2\varepsilon\le2^{42}\varepsilon.
 \label{eq:rs:infidelity}
\end{equation}
Discarding the residual systems and flags yields a state with squared fidelity at least $1-2^{42}\varepsilon$ with $\ket{\psistar}$. This proves the extraction bound for projective strategies, including mixed shared states because the purification environment was untouched.

For general POVMs, dilate all measurements on each party using a common embedding. For Alice, the map
\[
 \zeta\longmapsto\sum_{a=0}^{4}\sqrt{M_{a|x}}\,\zeta\otimes\ket a
\]
is an isometry for every question $x$. Extend it to a unitary $\mathcal U_x$ acting on inputs $\zeta\otimes\ket0$, and set
\[
 \widehat A_{xa}=\mathcal U_x^{\dagger}(I\otimes\ket a\bra a)\mathcal U_x,
 \qquad J_A\zeta=\zeta\otimes\ket0.
\]
Then $J_A^{\dagger}\widehat A_{xa}J_A=M_{a|x}$ for every $x,a$, with the same $J_A$ for all questions. Construct Bob's dilation in the same way. The embedded state has the original correlation and score, so the projective bound applies. Composing the extraction maps with these embeddings gives local channels on the original devices. Their extraction error is at most $2^{21}\sqrt\varepsilon$, and is always at most one, proving Theorem~\ref{thm:state-robust}.

\section{Identification and robustness of the measurements}
\label{sec:measurements}

The state estimate in Section~\ref{sec:robust-state} leaves two tasks: identify every labeled measurement effect at perfect score, and control all effect actions under the same isometries at near-perfect score. We first classify perfect strategies on their local reduced-state supports. We then show that the fixed extraction circuit realizes this classification. Finally, we extend its exact identities to commuting representations and use compactness to obtain a dimension-independent robustness modulus.

\subsection{Exact identification of the labeled effects}
\label{sec:measurement-exact}

The reconstruction rests on a simple observation. Two effects $E_-,E_+$ are determined by their sum $S=E_-+E_+$ and difference $R=E_+-E_-$:
\begin{equation}
 E_-=(S-R)/2,\qquad E_+=(S+R)/2.
 \label{eq:binary-recovery}
\end{equation}
The coarse agreements determine the pair sums, while the remaining forbidden outcomes determine the differences. Keeping their signs fixes the answer labels as well as the measurement bases.

For a perfect strategy, first compress the POVM effects to the local reduced-state supports and write them as $A_{xa},B_{yb}$. On these supports, a local operator annihilating the state must vanish. To see how the forbidden outcomes constrain the effects, suppose opposite-party effects $0\leq E,F\leq I$ have zero probability of mismatched outcomes. Since they commute,
\begin{equation}
 \begin{split}
  0&=\langle\Omega,[E(I-F)+(I-E)F]\Omega\rangle\\
   &=\norm{(E-F)\Omega}^{2}
      +\langle\Omega,(E-E^2+F-F^2)\Omega\rangle.
 \end{split}
 \label{eq:exact-agreement}
\end{equation}
Every term is nonnegative. Faithfulness of the local reduced states therefore gives $E^2=E$, $F^2=F$, and $E\Omega=F\Omega$. Thus perfect agreement yields both projectivity and an identity on the state, even when the original measurements are POVMs.

Applying this identity to the coarse agreements gives five matching coordinate subspaces. The relations among the first four make their pair differences purely off-diagonal, with unitary blocks identifying them with common residual spaces $\mathcal K_A,\mathcal K_B$. Lemma~\ref{lem:four-components} carries out this step and gives
\begin{equation}
 \Omega=\sum_{i=0}^{3}\ket{ii}\eta+\phi,
 \label{eq:measurement-four-components}
\end{equation}
where $\eta\in\mathcal K_A\otimes\mathcal K_B\otimes\mathcal H_E$ and $\phi\in\mathcal L_A\otimes\mathcal L_B\otimes\mathcal H_E$ occupies the fifth subspaces. The lemma also determines Alice's effects within the first four coordinates. Bob's coarse agreements determine his singleton effects and the two remaining pair sums, as recorded in~\eqref{eq:bob-pair-supports}.

Only the pairs with answers $2,3$ at $x=2,3$ and $y=1,2$ remain. The calculation for these pairs fixes both the missing effects and the amplitude ratio. Identify each pair's high coordinate with $\mathcal K_X$, keeping $\mathcal L_X$ fixed. In these common pair coordinates, write the second-minus-first differences as $R,R'$ for Alice and $C,D'$ for Bob. Projecting~\eqref{eq:measurement-four-components} onto the four pair supports gives $z_-,z_-,z_+,z_-$ at $(x,y)=(2,1),(3,1),(2,2),(3,2)$, where $z_\pm=\pm\eta/2+\phi$. The coefficients follow from the overlaps of $e_0,e_2$ with Bob's lines $h,k$ in Table~\ref{tab:bob}.

The local marginals of $z_\pm$ are direct sums of one quarter of the marginal of $\eta$ and the marginal of $\phi$, hence are faithful on the pair spaces. Apply~\eqref{eq:exact-agreement} to the three agreement tests and one disagreement test in Table~\ref{tab:low}, complementing one binary effect in the latter. The four differences are reflections, and
\begin{equation}
 \begin{aligned}
  Rz_-&=Cz_-,& R'z_-&=Cz_-,\\
  Rz_+&=-D'z_+,& R'z_-&=D'z_-.
 \end{aligned}
 \label{eq:four-binary-relations}
\end{equation}
The first two relations and faithfulness give $R=R'$; the first and fourth then give $C=D'$. Let $J_X=(-I_{\mathcal K_X})\oplus I_{\mathcal L_X}$. Since $J_Az_+=J_Bz_+=z_-$, cross-party commutation gives
\begin{equation}
 RJ_Az_+=CJ_Az_+=J_ACz_+=-J_ARz_+.
 \label{eq:measurement-anticommutation}
\end{equation}
Faithfulness turns this into $\{R,J_A\}=0$, and the same argument gives $\{C,J_B\}=0$. The diagonal blocks therefore vanish. Self-adjointness and $R^2=C^2=I$ now force
\begin{equation}
 R=R'=\begin{pmatrix}0&U_A\\U_A^{\dagger}&0\end{pmatrix},
 \qquad
 C=D'=\begin{pmatrix}0&U_B\\U_B^{\dagger}&0\end{pmatrix},
 \label{eq:measurement-binary-blocks}
\end{equation}
where $U_X:\mathcal L_X\to\mathcal K_X$ is unitary. Taking the high-Alice, low-Bob component of $Rz_-=Cz_-$ yields
\begin{equation}
 \begin{aligned}
  (U_A\otimes I\otimes I_E)\phi
    &=-\tfrac12(I\otimes U_B^\dagger\otimes I_E)\eta,\\
  (U_A\otimes(-U_B)\otimes I_E)\phi&=\eta/2.
 \end{aligned}
 \label{eq:measurement-amplitude}
\end{equation}
This proves the claimed amplitude ratio and fixes the relative sign. Appendix~\ref{app:exact-low} verifies the pair supports and conditional vectors used in this calculation.

Map the fifth subspace by $U_A$ for Alice and by $-U_B$ for Bob. This extracts the target state and makes the remaining differences positive swaps for Alice and negative swaps for Bob. For example, Alice's pair at $x=2$ has sum $(\ket0\bra0+\ket4\bra4)\otimes I$ and difference $(\ket0\bra4+\ket4\bra0)\otimes I$, so~\eqref{eq:binary-recovery} gives the labeled projections onto $(e_0-e_4)/\sqrt2$ and $(e_0+e_4)/\sqrt2$. The same calculation recovers all 20 Alice and 15 Bob effects. The equalities $R=R'$ and $C=D'$ ensure that one choice of local coordinates works for both questions.

Appendix~\ref{sec:proof} gives the full classification. It also returns from compressed effects to the original POVMs: because every compressed effect is a projection, positivity forces the original effect to preserve the reduced-state support, by~\eqref{eq:no-leakage}. Extending the coordinate identifications to local isometries proves Theorem~\ref{thm:strategy-exact}, including mixed states and joint effect actions.

\subsection{Realizing the classification with the fixed circuit}
\label{sec:measurement-circuit}

The exact classification constructs local coordinates from the occupied subspaces. For robustness we need a single extraction prescription that is also defined away from perfect score. The isometries $\mathcal V_A,\mathcal V_B$ in~\eqref{eq:rs:isometries} provide this prescription: their entries are fixed polynomials in the measurement projections, and they remain isometries at every score.

For projective strategies, Lemma~\ref{lem:rm:fixed-circuit} verifies that, at perfect score, the controlled unitaries implement precisely the coordinate identifications above, including Bob's minus sign on the fifth subspace. Here is how those identifications imply the measurement identities. On either local support, suppress the party index and write $v_i$ for the circuit coefficient carrying coordinate $i$ to coordinate zero. The operators $e_{ij}=v_i^\dagger v_j$ are matrix units and obey $v_ke_{ij}=\delta_{ki}v_j$. Thus, with the fixed auxiliary flag suppressed,
\begin{equation}
 \mathcal V_Xe_{ij}
  =\sum_k\ket k\otimes v_ke_{ij}
  =\ket i\otimes v_j
  =(\ket i\bra j\otimes I)\mathcal V_X.
 \label{eq:measurement-circuit-units}
\end{equation}
Every classified effect has coefficients given by its ideal matrix in this matrix-unit basis. Linearity of~\eqref{eq:measurement-circuit-units} therefore gives $\mathcal V_AA_{xa}=(A^\star_{xa}\otimes I)\mathcal V_A$, and likewise for Bob. The fixed circuit both extracts the state and intertwines all effects with the ideal measurements.

\Needspace{8\baselineskip}
\subsection{A common robustness modulus}
\label{sec:measurement-compactness}

Exact self-testing need not imply robustness~\cite{MancinskaSchmidt}. Moreover, finite-dimensional device spaces have no common dimension bound, so compactness cannot be applied directly to their strategies. We instead work with the state space of the universal algebra of two commuting projective measurement families. This space is compact and contains the limits needed for a dimension-independent argument.

Write $\mathfrak A$ for this universal algebra, with projection generators $a_{xa},b_{yb}$, and set
\[
 h_{\mathrm{loss}}=\frac1{12}\sum_{\mathsf w(x,y,a,b)=0}a_{xa}b_{yb}.
\]
A strategy induces a state $\omega$ on $\mathfrak A$ whose loss is $\omega(h_{\mathrm{loss}})$. The essential step is to prove that the fixed circuit's exact state and effect identities hold in every zero-loss commuting representation, including infinite-dimensional ones. Appendix~\ref{app:rm:commuting} establishes this by repeating the classification on local support corners, where the state is faithful. Thus every zero-loss state in the compact space satisfies the required identities.

\Needspace{9\baselineskip}
The following calculation converts these exact identities into simultaneous error control. Let $\Psi_\star=\ket{\psistar}\bra{\psistar}$ and let $\mathcal V$ be the joint polynomial column of the fixed circuit, so $\mathcal V^\dagger\mathcal V=I$. Ideal operators act on the extracted registers; identities on residual, flag, and purification spaces are suppressed. Define
\begin{align}
 E^A_{xa}&=\mathcal V a_{xa}-(A^\star_{xa}\otimes I)\mathcal V,\nonumber\\
 E^B_{yb}&=\mathcal V b_{yb}-(I\otimes B^\star_{yb}\otimes I)\mathcal V,\nonumber\\
 L&=\mathcal V^\dagger[(I-\Psi_\star)\otimes I]\mathcal V
    +\sum_{x,a}(E^A_{xa})^\dagger E^A_{xa}
    +\sum_{y,b}(E^B_{yb})^\dagger E^B_{yb}.
 \label{eq:rm:error-element}
\end{align}
This is one positive element of $\mathfrak A$, combining the extracted-state infidelity and all 35 squared effect errors. Lemma~\ref{lem:rm:perfect-commuting} gives $\omega(L)=0$ whenever $\omega(h_{\mathrm{loss}})=0$. On the state space $S(\mathfrak A)$, put
\begin{equation}
 g(\varepsilon)=\max\{\omega(L):\omega\in S(\mathfrak A),\ 
                                      \omega(h_{\mathrm{loss}})\leq\varepsilon\}.
 \label{eq:rm:modulus}
\end{equation}
The feasible sets are nonempty and weak-$*$ compact, so the maximum exists. Moreover,
\begin{equation}
 g(0)=0,\qquad \lim_{\varepsilon\downarrow0}g(\varepsilon)=0.
 \label{eq:rm:modulus-limit}
\end{equation}
Otherwise, a family of states with loss tending to zero and $L$-expectation bounded below by a positive constant would have a weak-$*$ convergent subnet. Its limit would have zero loss but positive $L$-expectation, contradicting the exact identities.

For a projective strategy of loss at most $\varepsilon$ and any purification, let $\xi=\mathcal V\Omega$ and $p=\langle\xi,(\Psi_\star\otimes I)\xi\rangle$. When $g(\varepsilon)<1$, we have $p\geq1-g(\varepsilon)>0$ and may choose the normalized vector $\chi=p^{-1/2}(\bra{\psistar}\otimes I)\xi$. With $\Phi=\ket{\psistar}\otimes\chi$, the two contributions to an Alice effect error satisfy
\begin{align}
 \norm{\xi-\Phi}&=\sqrt{2-2\sqrt p}\leq\sqrt{2g(\varepsilon)},
       \label{eq:rm:projective-errors}\\
 \norm{\mathcal V a_{xa}\Omega-(A^\star_{xa}\otimes I)\Phi}
 &\leq\norm{E^A_{xa}\Omega}+\norm{\xi-\Phi}
 \leq(1+\sqrt2)\sqrt{g(\varepsilon)}.
       \label{eq:measurement-common-error}
\end{align}
The same bound holds for every Bob effect. Thus one residual vector and the fixed isometries control the state and all single-effect actions at once. For larger errors the trivial bound two suffices. No dimension enters the definition of $g$.

The extension to mixed states and general POVMs is a direct application of the assumption-lifting theorem of Baptista et al.~\cite[Theorem 4.1(a)]{BaptistaEtAl}. Its hypotheses hold because the target is pure and has full Schmidt rank. The theorem gives common local isometries for the state and all single-effect actions; cross-party commutation then controls their joint actions. Appendix~\ref{app:rm:povm} completes this step, proving Theorem~\ref{thm:strategy-robust} without an explicit measurement-error rate.

\section{Discussion}
\label{sec:discussion}

The five-dimensional reduction retains the feature that makes Lalonde's game a counterexample to the sufficiency of maximal entanglement. Perfect play determines a fixed pure-state factor with a nonuniform Schmidt spectrum, together with every labeled measurement effect. The proof explains this rigidity through four equal components and one family of binary comparisons fixing the remaining amplitude. One of the original game's two low-amplitude coordinates can therefore be removed while retaining certification of nonmaximal entanglement.

The two robustness conclusions leave different quantitative questions open. The explicit state bound has a large prefactor, while the measurement bound gives convergence without a rate. Obtaining useful quantitative bounds for all measurement actions under the same isometries remains a natural next step. Further reductions of the state dimension or question and answer counts also require separate arguments; we make no minimality claim.

\newpage
\appendix
\section{Game definition and reduction}
\label{app:game}

We first specify the reference bases, then describe their reduction from Lalonde's game. The final subsection collects the support relations used in both the quantitative extraction and the exact measurement classification.

\subsection{The five-dimensional bases}

The game in Section~\ref{sec:setting} uses the following seven bases, obtained from Lalonde's six-dimensional bases by the reduction below. The question and answer labels in the tables are those used throughout the proof.

Let $e_0,\ldots,e_4$ be the standard basis of $\R^5$ and set
\begin{equation}
  D=\diag(2,2,2,2,1).
  \label{eq:D}
\end{equation}
The following unit vectors lie in the first four coordinates:
\begin{equation}
\begin{aligned}
  c_0&=(-e_0-e_1-e_2+e_3)/2,&c_1&=(e_0+e_1-e_2+e_3)/2,\\
  h&=(-e_0+e_1-e_2-e_3)/2,&h'&=(e_0-e_1-e_2-e_3)/2,\\
  d_0&=(-e_0+e_1-e_2+e_3)/2,&d_1&=(e_0+e_1+e_2+e_3)/2,\\
  k&=(e_0+e_1-e_2-e_3)/2,&k'&=(-e_0+e_1+e_2-e_3)/2.
\end{aligned}
\label{eq:lines}
\end{equation}
The two lists $(c_0,c_1,h,h')$ and $(d_0,d_1,k,k')$ are orthonormal bases of $\operatorname{span}\{e_0,e_1,e_2,e_3\}$. Tables~\ref{tab:alice} and~\ref{tab:bob} specify the game. Normalize each vector in a row separately, and denote the resulting vectors by $u_{xa}$ and $v_{yb}$. Each row is then an orthonormal basis of $\R^5$.

\begin{table}[htbp]
\centering
\renewcommand{\arraystretch}{1.22}
\begin{tabular}{c@{\quad}ccccc}
\toprule
$x\backslash a$&0&1&2&3&4\\
\midrule
0&$e_0-e_1$&$e_0+e_1$&$e_2-e_3$&$e_2+e_3$&$e_4$\\
1&$e_4$&$e_0-e_2$&$e_0+e_2$&$e_1-e_3$&$e_1+e_3$\\
2&$e_1-e_2$&$e_1+e_2$&$e_0-e_4$&$e_0+e_4$&$e_3$\\
3&$e_0-e_3$&$e_0+e_3$&$e_2-e_4$&$e_2+e_4$&$e_1$\\
\bottomrule
\end{tabular}
\caption{Alice's four bases. Answer labels follow the columns. Each displayed vector is normalized individually.}
\label{tab:alice}
\end{table}

\begin{table}[htbp]
\centering
\renewcommand{\arraystretch}{1.22}
\begin{tabular}{c@{\quad}ccccc}
\toprule
$y\backslash b$&0&1&2&3&4\\
\midrule
0&$e_0$&$e_1$&$e_2$&$e_3$&$e_4$\\
1&$c_0$&$c_1$&$h+e_4$&$h-e_4$&$h'$\\
2&$d_0$&$d_1$&$k+e_4$&$k-e_4$&$k'$\\
\bottomrule
\end{tabular}
\caption{Bob's three bases, with the same normalization convention. The vectors in the last two rows are defined in~\eqref{eq:lines}.}
\label{tab:bob}
\end{table}

\subsection{Reduction of Lalonde's six-dimensional game}
\label{sec:reduction}

Write $G_6$ for Lalonde's game $G_{4,3,6,6}$, with the questions and answers in~\cite[Appendix A]{Lalonde} reindexed from zero. To make the relation to $G_5$ explicit, let $f_0,\ldots,f_5$ be the coordinates used there. The original Schmidt matrix is $S_6=\diag(1,1,2,2,2,2)$. Define a real orthogonal matrix $O$ by its columns:
\begin{equation}
 O=\frac1{\sqrt2}\bigl(f_4-f_5,\ f_4+f_5,\ f_2-f_3,\ f_2+f_3,\ f_0-f_1,\ f_0+f_1\bigr).
 \label{eq:red:coordinates}
\end{equation}
This is Bob's first basis matrix in the original game. Applying $O^{\mathsf T}$ on both sides gives
\begin{equation}
 D_6=O^{\mathsf T}S_6O=\diag(2,2,2,2,1,1),\qquad
 \ket{\psi_6}=\frac{2\sum_{i=0}^3\ket{ii}+\ket{44}+\ket{55}}{\sqrt{18}}.
 \label{eq:red:state-six}
\end{equation}
We use $e_0,\ldots,e_5$ for these new coordinates, identifying their first five directions with the five-dimensional coordinates above.

In these coordinates, Lalonde's seven bases have the following representatives, in their original answer order. Normalize every displayed vector separately; signs of individual vectors do not affect the measurement projections or the predicate:
\begin{equation}
\begin{aligned}
 \mathcal U^{(6)}_0&=(e_0-e_1,e_0+e_1,e_2-e_3,e_2+e_3,e_4,e_5),\\
 \mathcal U^{(6)}_1&=(e_4,e_5,e_0-e_2,e_0+e_2,e_1-e_3,e_1+e_3),\\
 \mathcal U^{(6)}_2&=(e_1-e_2,e_1+e_2,e_0-e_4,e_0+e_4,e_3-e_5,e_3+e_5),\\
 \mathcal U^{(6)}_3&=(e_0-e_3,e_0+e_3,e_2-e_4,e_2+e_4,e_1-e_5,e_1+e_5),\\
 \mathcal V^{(6)}_0&=(e_0,e_1,e_2,e_3,e_4,e_5),\\
 \mathcal V^{(6)}_1&=(c_0,c_1,h+e_4,h-e_4,h'+e_5,h'-e_5),\\
 \mathcal V^{(6)}_2&=(d_0,d_1,k+e_4,k-e_4,k'+e_5,k'-e_5).
\end{aligned}
\label{eq:red:bases-six}
\end{equation}
Here the vectors from~\eqref{eq:lines} are embedded in $\R^6$ with zero final coordinate. Equation~\eqref{eq:red:bases-six} follows by multiplying the original basis matrices by $O^{\mathsf T}$, without changing any question or answer order.

Let $P=I-e_5e_5^{\mathsf T}$ and identify $P\R^6$ with $\R^5$. The reduction compresses each reference effect to this subspace. For a unit vector $r$ in the first five coordinates, the rank-one projection $\pi(v)=vv^{\dagger}/\norm v^2$ satisfies
\begin{equation}
 P\pi(e_5)P=0,\qquad
 P\pi(r\pm e_5)P=\tfrac12\pi(r),\qquad
 P\bigl[\pi(r+e_5)+\pi(r-e_5)\bigr]P=\pi(r).
 \label{eq:red:compression}
\end{equation}
Thus an outcome supported on $e_5$ disappears, while a pair $r\pm e_5$ becomes a single outcome supported on $r$. Alice's last pair becomes $e_3$ at $x=2$ and $e_1$ at $x=3$; Bob's last pair becomes $h'$ at $y=1$ and $k'$ at $y=2$. The other directions already lie in $P\R^6$ and are unchanged. These operations give exactly Tables~\ref{tab:alice} and~\ref{tab:bob}.

For completeness, Table~\ref{tab:red:labels} records which old outcomes contribute to each new outcome. Write these sets as $F^A_x(a)$ and $F^B_y(b)$. If $A^{(6)}_{x\alpha},B^{(6)}_{y\beta}$ are the original reference projections in the coordinates above, then
\begin{equation}
 A^\star_{xa}=\left.\sum_{\alpha\in F^A_x(a)}PA^{(6)}_{x\alpha}P\right|_{P\R^6},\qquad
 B^\star_{yb}=\left.\sum_{\beta\in F^B_y(b)}PB^{(6)}_{y\beta}P\right|_{P\R^6}.
 \label{eq:red:effect-map}
\end{equation}
\begin{table}[htbp]
\centering
\renewcommand{\arraystretch}{1.15}
\begin{tabular}{cc@{\quad}ccccc}
\toprule
Party&Question&New 0&New 1&New 2&New 3&New 4\\
\midrule
Alice&0&$\{0\}$&$\{1\}$&$\{2\}$&$\{3\}$&$\{4\}$\\
Alice&1&$\{0\}$&$\{2\}$&$\{3\}$&$\{4\}$&$\{5\}$\\
Alice&2,3&$\{0\}$&$\{1\}$&$\{2\}$&$\{3\}$&$\{4,5\}$\\
Bob&0&$\{0\}$&$\{1\}$&$\{2\}$&$\{3\}$&$\{4\}$\\
Bob&1,2&$\{0\}$&$\{1\}$&$\{2\}$&$\{3\}$&$\{4,5\}$\\
\bottomrule
\end{tabular}
\caption{Old answer labels contributing to each answer of the reduced game. The omitted old answers are Alice's answer 5 at question 0, her answer 1 at question 1, and Bob's answer 5 at question 0.}
\label{tab:red:labels}
\end{table}

The state is reduced on the same subspace:
\begin{equation}
 (P\otimes P)\ket{\psi_6}=\sqrt{\frac{17}{18}}\ket{\psistar},\qquad
 D=\left.PD_6P\right|_{P\R^6}.
 \label{eq:red:state-map}
\end{equation}
This identity describes the relation between the two reference states. The self-testing extraction maps in Section~\ref{sec:robust-state} are isometries defined on arbitrary strategies for $G_5$.

The winning predicate has an equally explicit reduction. Let $\mathsf w_6$ denote Lalonde's original predicate. Then
\begin{equation}
 \mathsf w(x,y,a,b)=
 \max_{\substack{\alpha\in F^A_x(a)\\\beta\in F^B_y(b)}}\mathsf w_6(x,y,\alpha,\beta).
 \label{eq:red:predicate}
\end{equation}
In words, a pair of new answers is accepted if at least one pair of its old representatives was accepted. To verify~\eqref{eq:red:predicate}, first suppose at least one of the two sets of old labels is a singleton. Its basis vector has no $e_5$ component, so every relevant old weighted inner product is a nonzero scalar multiple of the new one. If both sets are $\{4,5\}$, the old directions are $r\pm e_5$ and $s\pm e_5$, where $r\in\{e_3,e_1\}$ and $s\in\{h',k'\}$. Their unnormalized weighted inner products are $2\ip r s\pm1$. In all four cases $|\ip r s|=1/2$, so some old pair is accepted, and the reduced weighted inner product $2\ip r s\in\{-1,1\}$ is also nonzero. This proves the identity for every question and answer pair.

Equations~\eqref{eq:red:compression}--\eqref{eq:red:predicate} specify the sense in which $G_5$ reduces Lalonde's game: they remove one reference coordinate and reduce each answer alphabet, with all question labels retained. Deleting and identifying answers can change which perfect strategies exist, so the properties of $G_5$ require their own proof. The theorems in Section~\ref{sec:setting} establish these properties for the reduced predicate.

\FloatBarrier
\subsection{Support relations used in the proof}
\label{app:support}

The predicate~\eqref{eq:game} determines all 300 entries of the game. The following tables organize the forbidden outcomes by their role in the proof: matching the five labels, aligning the first four components, and fixing the fifth amplitude. Appendix~\ref{app:robust-transport} uses their quantitative consequences; Appendix~\ref{sec:proof} uses the exact relations to classify the measurements. Every entry follows directly from the reference vectors and their weighted inner products.

For a set of answers $S$ or $T$, write
\[
 A_x(S)=\sum_{a\in S}A_{xa},\qquad B_y(T)=\sum_{b\in T}B_{yb}.
\]
Each row of Table~\ref{tab:coarse} means that an answer pair is forbidden whenever membership in $S$ and $T$ differs. Perfect play therefore gives an agreement relation between $A_x(S)$ and $B_y(T)$. These relations supply the coordinate projections and the two commuting pairs of coarse projections used in the high-component argument.

\begin{table}[htbp]
\centering
\renewcommand{\arraystretch}{1.12}
\begin{tabular}{cccc}
\toprule
$x$&$y$&$S$&$T$\\
\midrule
0&0&$\{0,1\}$&$\{0,1\}$\\
0&0&$\{2,3\}$&$\{2,3\}$\\
0&0&$\{4\}$&$\{4\}$\\
1&0&$\{1,2\}$&$\{0,2\}$\\
1&0&$\{3,4\}$&$\{1,3\}$\\
1&0&$\{0\}$&$\{4\}$\\
2&0&$\{0,1\}$&$\{1,2\}$\\
2&0&$\{2,3\}$&$\{0,4\}$\\
2&0&$\{4\}$&$\{3\}$\\
3&0&$\{0,1\}$&$\{0,3\}$\\
3&0&$\{2,3\}$&$\{2,4\}$\\
3&0&$\{4\}$&$\{1\}$\\
\midrule
0&1&$\{1,2\}$&$\{0,1\}$\\
1&1&$\{0,2,3\}$&$\{0,2,3\}$\\
0&2&$\{0,2\}$&$\{0,4\}$\\
1&2&$\{2,4\}$&$\{0,1\}$\\
\bottomrule
\end{tabular}
\caption{Coarse agreements. Every mismatch between the two displayed membership conditions is rejected.}
\label{tab:coarse}
\end{table}

Table~\ref{tab:high} gives the complete set of accepted Bob answers for each indicated Alice answer. Its complementary forbidden events fix the products of the high-component reflections.

\begin{table}[htbp]
\centering
\renewcommand{\arraystretch}{1.12}
\begin{tabular}{cccc}
\toprule
$x$&$y$&$a$&Accepted $b$\\
\midrule
2&1&0&$\{1,2,3\}$\\
2&1&1&$\{0,4\}$\\
2&2&0&$\{0,2,3\}$\\
2&2&1&$\{1,4\}$\\
3&1&0&$\{0,4\}$\\
3&1&1&$\{1,2,3\}$\\
3&2&0&$\{0,2,3\}$\\
3&2&1&$\{1,4\}$\\
\bottomrule
\end{tabular}
\caption{Support restrictions for Alice's high-component pairs. Answers outside the last column are rejected.}
\label{tab:high}
\end{table}

Table~\ref{tab:low} restricts both parties to the answer pair $\{2,3\}$. Three rows enforce agreement of the binary answers; the remaining row enforces disagreement. These are conditional restrictions within the indicated pair supports, not agreement relations for the full measurements.

\begin{table}[htbp]
\centering
\renewcommand{\arraystretch}{1.2}
\begin{tabular}{cccl}
\toprule
$x$&$y$&Forbidden $(a,b)$ in $\{2,3\}^2$&Relation\\
\midrule
2&1&$(2,3),(3,2)$&Agreement\\
3&1&$(2,3),(3,2)$&Agreement\\
2&2&$(2,2),(3,3)$&Disagreement\\
3&2&$(2,3),(3,2)$&Agreement\\
\bottomrule
\end{tabular}
\caption{The four binary restrictions fixing the fifth Schmidt amplitude.}
\label{tab:low}
\end{table}

One additional support restriction is used in the quantitative argument: at $x=0$ and $y\in\{1,2\}$, Alice's answer $4$ forbids every Bob answer outside $\{2,3\}$. It gives the low-component leakage bound~\eqref{eq:rt:low-bob-pair-leakage}. In the exact classification, Bob's pair supports are identified from the coarse projections in Table~\ref{tab:coarse}.

\FloatBarrier
\Needspace{10\baselineskip}
\section{Quantitative estimates for state extraction}
\label{app:robust-transport}

This appendix proves the three estimates used in Section~\ref{sec:robust-state}: the label error~\eqref{eq:rs:label-error}, the high-branch error~\eqref{eq:rs:high-branch-bounds}, and the fifth-branch error~\eqref{eq:rs:branch-bounds}. We first define the complete extraction maps and control the labels. We then show that the high-component reflections exchange the label subspaces. Finally, we put the four binary tests in common coordinates and use them to compare the fifth branch with one half of the first. All estimates are on unnormalized vectors.

\subsection{The extraction maps and matching labels}
\label{app:sequential-labels}

The construction keeps every branch of Alice's sequential measurement, so the maps are isometries for every projective strategy.

Use the projective measurement notation $A_{xa},B_{yb}$ and the fixed purification $\Omega$ from Section~\ref{sec:robust-state}. The loss assumption is
\begin{equation}
 \sum_{\mathsf w(x,y,a,b)=0}\norm{A_{xa}B_{yb}\Omega}^{2}
 \le12\varepsilon=\delta^2,\qquad\delta=\sqrt{12\varepsilon}.
 \label{eq:rs:loss}
\end{equation}
The sum of the conditional probabilities of any collection of distinct forbidden events is at most $\delta^2$. In particular, every coarse agreement in Table~\ref{tab:coarse} gives
\begin{equation}
 \norm{(P-Q)\Omega}^2
 =\langle\Omega,[P(I-Q)+(I-P)Q]\Omega\rangle\le\delta^2,
 \label{eq:rs:agreement}
\end{equation}
where $P,Q$ are the indicated opposite-party coarse projections. The corresponding reflections differ on $\Omega$ by at most $2\delta$.

Bob's first measurement supplies the five labels: set $Q_i=B_{0i}$. On Alice's side, define
\begin{equation}
  U=A_{00}+A_{01},\qquad W=A_{02}+A_{03},\qquad
  V=A_{11}+A_{12},\qquad T=A_{13}+A_{14}.
  \label{eq:coarse-four}
\end{equation}
Alice's five intended branches and two additional branches are
\begin{equation}
 (K_0,K_1,K_2,K_3,K_4)=(VU,TU,VW,TW,A_{04}),
 \qquad (L_0,L_1)=(A_{10}U,A_{10}W).
 \label{eq:rs:kraus}
\end{equation}
At imperfect score the products need not be projections. They nevertheless satisfy
\begin{equation}
 \sum_{i=0}^{4}K_i^{\dagger}K_i+\sum_{j=0}^{1}L_j^{\dagger}L_j=I,
 \label{eq:rs:complete}
\end{equation}
because $V+T+A_{10}=I$ and $U+W+A_{04}=I$.

Define the reflections
\begin{equation}
\begin{aligned}
 f&=2(A_{01}+A_{02})-I,&F&=2(B_{10}+B_{11})-I,\\
 g&=2(A_{10}+A_{12}+A_{13})-I,&G&=2(B_{10}+B_{12}+B_{13})-I.
\end{aligned}
\label{eq:rs:reflections}
\end{equation}
The Bob reflections commute, and~\eqref{eq:rs:agreement} gives
$\norm{(f-F)\Omega},\norm{(g-G)\Omega}\le2\delta$.
Use the controlled unitaries
\begin{equation}
 (T^A_0,T^A_1,T^A_2,T^A_3)=(I,f,g,-fg),\qquad
 (T^B_0,T^B_1,T^B_2,T^B_3)=(I,F,G,-FG).
 \label{eq:rs:high-unitaries}
\end{equation}
For the fifth branch, we need two unitary changes of coordinates on Bob's space. For $v=h,k$ from~\eqref{eq:lines}, write its first four coordinates as $v=(s_0,s_1,s_2,s_3)/2$, where $s_i\in\{\pm1\}$, and set
\begin{equation}
  \begin{split}
    E(L) & =(I+\mathrm iL)/\sqrt2,\qquad L\text{ a reflection}, \\
    D_v & =s_0Q_0-\mathrm i s_1Q_1-\mathrm i s_2Q_2+s_3Q_3 +Q_4, \\
    V_v & =[D_vE(F)E(G)]^{\dagger}.
  \end{split}
  \label{eq:rt:charts}
\end{equation}
Every factor in $V_v$ is unitary on the full local space.
Define
\begin{equation}
 T^A_4=I-2A_{22},\qquad
 T^B_4=-V_h(I-2B_{12})V_h^{\dagger}.
 \label{eq:rs:low-unitaries}
\end{equation}
These are unitaries at every score, with polynomial entries in the measurement projections. We will show that, at perfect score, $V_h$ carries the occupied high line $h$ to coordinate zero and fixes the occupied low subspace. Its action on the high and low components is estimated in Appendices~\ref{app:coordinate-evaluation} and~\ref{app:conditional-vectors}, respectively.

Let $\ket{\mathrm g},\ket{\mathrm b_0},\ket{\mathrm b_1}$ be orthonormal local flag vectors. Define
\begin{equation}
\begin{aligned}
 \mathcal V_A\zeta
 &=\sum_{i=0}^{4}\ket i\otimes T_i^AK_i\zeta\otimes\ket{\mathrm g}
   +\sum_{j=0}^{1}\ket0\otimes L_j\zeta\otimes\ket{\mathrm b_j},\\
 \mathcal V_B\zeta
 &=\sum_{i=0}^{4}\ket i\otimes T_i^BQ_i\zeta.
\end{aligned}
\label{eq:rs:isometries}
\end{equation}
Both maps are isometries by~\eqref{eq:rs:complete} and $\sum_iQ_i=I$. The additional branches are retained in orthogonal flags, so the construction does not postselect. The output labels follow the reference coordinates in Tables~\ref{tab:alice} and~\ref{tab:bob}; Lemma~\ref{lem:rm:fixed-circuit} will verify the exact measurement actions in these coordinates.

\begin{lemma}
\label{lem:rs:labels}
For every projective strategy,
\begin{equation}
 \sum_{i=0}^{4}\norm{(K_i-Q_i)\Omega}^{2}
 +\sum_{j=0}^{1}\norm{L_j\Omega}^{2}\le4\delta^2.
 \label{eq:rs:label-distance}
\end{equation}
\end{lemma}
\begin{proof}
Consider the projective partitions
\[
 \mathcal P=(U,W,A_{04}),\quad
 \mathcal Q=(Q_0+Q_1,Q_2+Q_3,Q_4),\quad
 \mathcal R=(V,T,A_{10}),\quad
 \mathcal S=(Q_0+Q_2,Q_1+Q_3,Q_4).
\]
Their mismatches consist of distinct forbidden events at question pairs
$(0,0)$ and $(1,0)$. Orthogonality therefore gives
\[
 \sum_{a=0}^{2}\norm{(\mathcal P_a-\mathcal Q_a)\Omega}^{2}
 +\sum_{b=0}^{2}\norm{(\mathcal R_b-\mathcal S_b)\Omega}^{2}
 \le2\delta^2.
\]
For $a=0,1$ and $b=0,1,2$, cross-party commutation gives
\[
 (\mathcal R_b\mathcal P_a-\mathcal S_b\mathcal Q_a)\Omega
 =\mathcal R_b(\mathcal P_a-\mathcal Q_a)\Omega
  +\mathcal Q_a(\mathcal R_b-\mathcal S_b)\Omega.
\]
The six Bob intersections are $(Q_0,Q_1,0,Q_2,Q_3,0)$ in
lexicographic order. Thus these six differences, together with
$\mathcal P_2-\mathcal Q_2$, are exactly the five $K_i-Q_i$ and
the two $L_j$, up to reordering.
Apply $\norm{u+v}^2\le2\norm u^2+2\norm v^2$ to each displayed
difference and sum. Since $\sum_b\mathcal R_b=I$ and
$\mathcal Q_0+\mathcal Q_1\le I$, adding the final branch
bounds the desired sum by twice the preceding mismatch bound,
which is $4\delta^2$.
\end{proof}

Replace $K_i\Omega$ by $Q_i\Omega$ and $L_j\Omega$ by zero in Alice's isometry. Orthogonality of the labels and flags bounds the error by $2\delta$. Bob's isometry preserves this error. Because his labels are orthogonal and commute with Alice's controlled unitaries, the comparison vector is exactly the sum in~\eqref{eq:rs:label-error}.

\subsection{Relations controlling the first four label subspaces}
\label{app:high-relations}

Our next goal is to show that conjugation by $F$ and $G$ approximately exchanges the appropriate label projections. This will align the first four extracted branches. We begin with supported versions of these reflections and derive their square and product relations. These will imply approximate anticommutation with the two label signs. Throughout, the projective strategy and purification are those fixed above.

We repeatedly transfer a right factor to the other party before estimating an error. If $E,C$ act on Bob, $c$ acts on Alice, and
\[
 \norm{E\Omega}\le e\delta,\qquad \norm E\le M,
 \qquad \norm c\le m,\qquad \norm{(C-c)\Omega}\le d\delta,
\]
then cross-party commutation gives
\begin{equation}
 \norm{EC\Omega}\le (me+Md)\delta,
 \qquad EC\Omega=E(C-c)\Omega+cE\Omega.
 \label{eq:robust-transfer}
\end{equation}
The same rule holds with the parties exchanged. For products of contractions, opposite-party approximations occur in reverse order and their errors add. All estimates below concern unnormalized vectors; no lower bound on a branch probability is used.

In addition to the extraction reflections $f,g,F,G$, define
\begin{align*}
  \widehat F & =I-2(B_{20}+B_{24}), & \widehat G & =2(B_{20}+B_{21})-I, \\
  \widehat f & =I-2(A_{00}+A_{02}), & \widehat g & =2(A_{12}+A_{14})-I.
\end{align*}
They are reflections, $[F,G]=[\widehat F,\widehat G]=0$, and
\begin{equation}
  \lVert(F-f)\Omega\rVert,\quad \lVert(G-g)\Omega\rVert,\quad \lVert(\widehat F-\widehat f)\Omega\rVert,\quad \lVert(\widehat G-\widehat g)\Omega\rVert \leq 2\delta.
  \label{eq:rt:high-reflections}
\end{equation}
The last two inequalities follow from the coarse pairs $A_{00}+A_{02}$ and $B_{20}+B_{24}$, and $A_{12}+A_{14}$ and $B_{20}+B_{21}$, respectively. The complete mismatch set in each row of Table~\ref{tab:coarse} has total probability at most $\delta^2$, and hence bounds the corresponding projection disagreement by $\delta$.

Write $Q_{ij}=Q_i+Q_j$ and put
\[
  P=Q_{01},\quad W_B=Q_{23},\quad V_B=Q_{02},\quad T_B=Q_{13}, \qquad H=Q_0+Q_1+Q_2+Q_3.
\]
\begin{align*}
  R & =PFP+W_BFW_B, & S & =V_BGV_B+T_BGT_B, \\
  J & =P-W_B-Q_4, & K & =V_B-T_B-Q_4, \\
  R' & =JR, & S' & =KS.
\end{align*}
Here $H$ is an orthogonal projection, $J,K$ are reflections, and $R,S,R',S'$ are self-adjoint contractions supported on $H$: each such operator $L$ satisfies $L=HLH$. The signs of $J$ are constant on each pair supporting a block of $R$, and the signs of $K$ are constant on each pair supporting a block of $S$.

The Alice approximations to $P,W_B,V_B,T_B$ are $U,W,V,T$, respectively, with error $\delta$. Each $Q_i$, $0\leq i<4$, has an Alice contraction approximation of error $2\delta$, namely the corresponding member of $(VU,TU,VW,TW)$. The Alice approximations to $J,K$ are
\[
  j_A=U-W-A_{04}=2U-I,\qquad k_A=V-T-A_{10}=2V-I,
\]
with error $2\delta$. Indeed, $J=2P-I$ and $K=2V_B-I$, so the agreements between $U,P$ and between $V,V_B$ give these bounds directly. Within Alice's individual measurements, $j_Af=fj_A=\widehat f$ and $k_Ag=gk_A=\widehat g$. Consequently
\begin{equation}
  \lVert(JF-\widehat F)\Omega\rVert,\qquad \lVert(KG-\widehat G)\Omega\rVert\leq6\delta.
  \label{eq:rt:high-signed-reflections}
\end{equation}

If $p\in\{P,W_B\}$ and $p_A\in\{U,W\}$ is its Alice approximation, then $[f,p_A]=0$ and cross-party commutation gives
\[
  [F,p]=[F-f,p-p_A].
\]
Using $\lVert F-f\rVert\leq2$ and $\lVert p-p_A\rVert\leq1$ to bound the two terms on the right gives $\lVert[F,p]\Omega\rVert\leq4\delta$. The same argument applies to $G$ and $p\in\{V_B,T_B\}$. The two summands in
\[
  D_R:=R-HF=P[F,P]+W_B[F,W_B]
\]
have orthogonal ranges, so $\lVert D_R\Omega\rVert\leq 4\sqrt2\delta\leq6\delta$. The analogous estimates are collected below:
\begin{equation}
  \begin{array}{c|c|c}
    \text{operator} & \text{operator norm bound} & \text{bound on its norm on }\Omega \\ \hline
    D_R=R-HF & 2 & 6\delta \\
    D_S=S-HG & 2 & 6\delta \\
    D_{R'}=R'-H\widehat F & 2 & 12\delta \\
    D_{S'}=S'-H\widehat G & 2 & 12\delta
  \end{array}
  \label{eq:rt:high-differences}
\end{equation}
The same argument bounds $D_S$. For the primed differences, use $D_{R'}=JD_R+H(JF-\widehat F)$ and its analogue for $S'$.

The identity $R^2-H=RD_R+D_RF$ and~\eqref{eq:robust-transfer} give
\[
  \lVert(R^2-H)\Omega\rVert\leq(6+6+4)\delta=16\delta.
\]
The same bound holds for $S^2-H$. Since $H-R^2$ and $H-S^2$ are positive contractions, right multiplication by any $Q_i$ with $0\le i<4$ adds at most $2\delta$. Every diagonal block of either square defect therefore has norm at most $18\delta$ when applied to $\Omega$.

Likewise,
\[
 RS-HFG=RD_S+D_RG
\]
has norm on $\Omega$ at most $16\delta$; exchanging $R,S$ gives
the same bound for $SR-HGF$. Using the primed differences in
\eqref{eq:rt:high-differences} gives $28\delta$ for each of
$R'S'-H\widehat F\widehat G$ and $S'R'-H\widehat G\widehat F$.

Two further relations follow from the forbidden outcomes. Set
\[
  N_1=\frac{I-FG}{2}=B_{11}+B_{12}+B_{13},\qquad N_2=\frac{I-\widehat F\widehat G}{2}=B_{20}+B_{22}+B_{23}.
\]
For $e=A_{20}$, $e'=A_{21}$ and $E=e+e'$, both terms in $e-EN_\ell=e(I-N_\ell)-e'N_\ell$ are forbidden, for $\ell=1,2$. They have orthogonal ranges and their squared norms sum to at most $\delta^2$. Hence $\lVert(e-EN_\ell)\Omega\rVert \leq\delta$. Subtracting the two relations gives an error $2\delta$. Replacing $E\Omega$ by $Q_{12}\Omega$ adds at most $\delta$, since $\lVert N_1-N_2\rVert\leq1$. It follows that
\begin{equation}
  \lVert(FG-\widehat F\widehat G)Q_{12}\Omega\rVert\leq6\delta.
  \label{eq:rt:high-first-product}
\end{equation}
For $A_{30},A_{31}$, the corresponding two projections are $I-N_1$ and $N_2$, and $A_{30}+A_{31}$ agrees with $Q_{03}$ on $\Omega$ up to $\delta$. The same calculation gives
\begin{equation}
  \lVert(FG+\widehat F\widehat G)Q_{03}\Omega\rVert\leq6\delta.
  \label{eq:rt:high-second-product}
\end{equation}
Both calculations use the forbidden rectangles in Table~\ref{tab:high}.

\subsection{Aligning the first four extracted branches}
\label{app:high-branch-alignment}
The product relations force the supported reflections to exchange the
appropriate labels. Put $D_H=JKH=Q_{03}-Q_{12}$.
The projection $H$ and the contraction $D_H$ have Alice contraction
approximations
\[
 I-A_{04},\qquad (A_{30}+A_{31})-(A_{20}+A_{21}),
\]
with errors $\delta$ and $2\delta$, respectively.
For the second approximation, a difference of two projections is a
contraction, even when the projections do not commute.
The two product comparisons above, followed by
\eqref{eq:rt:high-first-product}--\eqref{eq:rt:high-second-product}, give
\begin{align*}
 \norm{(RS+R'S'JK)\Omega}
 &\le\norm{(RS-HFG)H\Omega}
       +\norm{(R'S'-H\widehat F\widehat G)D_H\Omega}\\
 &\quad+\norm{H(FGH+\widehat F\widehat G D_H)\Omega}\\
 &\le(18+32+12)\delta=62\delta.
\end{align*}
The last term splits into the two tested subspaces $Q_{12}$ and $Q_{03}$.
The same estimate holds for $SR+S'R'JK$, using the product comparisons
in the reversed order and the commutation of $F,G$ and of
$\widehat F,\widehat G$.
Since $R$ commutes with $J$, $S$ commutes with $K$, and $J$ commutes
with $K$, the following are exact identities:
\[
 RS+R'S'JK=R(S+JSJ),\qquad
 SR+S'R'JK=S(R+KRK).
\]
Each of $R,S$ has an Alice contraction approximation with error
$9\delta$: use its $6\delta$ comparison with $HF$ or $HG$,
then transfer the two factors at cost $3\delta$.
The corresponding approximations to $KRK$ and $JSJ$ have error
$13\delta$. Since $H-R^2$ is a positive contraction whose norm on
$\Omega$ is at most $16\delta$, the identity
\[
 S+JSJ=R\bigl(RS+R'S'JK\bigr)+(H-R^2)S+(H-R^2)JSJ
\]
and~\eqref{eq:robust-transfer} imply
$\norm{(S+JSJ)\Omega}\le(62+25+29)\delta=116\delta$.
Interchanging the roles of the supported reflections gives
\begin{equation}
 \norm{(R+KRK)\Omega},\quad\norm{(S+JSJ)\Omega}\le116\delta.
 \label{eq:rt:high-anticommutation}
\end{equation}

For $i\in\{0,1,2,3\}$, let $i\oplus1$ and $i\oplus2$ denote
changing the last and first binary digits of $i$, respectively; set
$i\oplus3=(i\oplus1)\oplus2$.
Both $J$ and $K$ are scalar on each label subspace. Compressing
\eqref{eq:rt:high-anticommutation} by $Q_i$ therefore gives
\[
 \norm{Q_iRQ_i\Omega},\quad\norm{Q_iSQ_i\Omega}
 \le\tfrac12(116+4)\delta=60\delta,
\]
where right multiplication by $Q_i$ costs at most $4\delta$.
The $18\delta$ bounds on the compressed square defects then give
\begin{equation}
 \begin{aligned}
 \norm{(Q_iRQ_{i\oplus1}RQ_i-Q_i)\Omega}&\le78\delta,\\
 \norm{(Q_iSQ_{i\oplus2}SQ_i-Q_i)\Omega}&\le78\delta.
 \end{aligned}
 \label{eq:rt:high-diagonal-elimination}
\end{equation}
Indeed, the omitted term in each compressed square is the square of
its diagonal block, whose norm on $\Omega$ is at most $60\delta$.

\begin{lemma}[Conjugation of the label projections]
  \label{lem:rt:high-endpoints}
  For every $i\in\{0,1,2,3\}$,
  \begin{align}
    \lVert(FQ_iF-Q_{i\oplus1})\Omega\rVert & \leq400\delta, & \lVert(GQ_iG-Q_{i\oplus2})\Omega\rVert & \leq400\delta,
    \label{eq:rt:high-single-swaps}                                \\
    \lVert(FGQ_iGF-Q_{i\oplus3})\Omega\rVert & \leq800\delta.
    \label{eq:rt:high-double-swaps}
  \end{align}
\end{lemma}
\begin{proof}
  Consider $X=Q_0FQ_1$, $B=Q_1FQ_1$, and $L=FQ_1-X$. The leakage outside $P=Q_{01}$ obeys
  \[
    \lVert(I-P)FQ_1\Omega\rVert \leq\lVert[F,P]Q_1\Omega\rVert\leq(4+4)\delta.
  \]
  Thus $\lVert L\Omega\rVert\leq68\delta$ and $\lVert L\rVert\leq2$. Separately,
  \[
    L^{\dagger}=Q_1F(I-P)+B=Q_1[P,F]+B
  \]
  gives $\lVert L^{\dagger}\Omega\rVert\leq64\delta$. The identity
  \[
    FQ_1F-Q_0=LF+XL^{\dagger}+(XX^{\dagger}-Q_0)
  \]
  gives the bound $(72+64+78)\delta=214\delta$. The same calculation applies to every directed horizontal or vertical pair of label subspaces. For the composition of the two swaps,
  \[
    FGQ_iGF-Q_{i\oplus3} =G(FQ_iF-Q_{i\oplus1})G +(GQ_{i\oplus1}G-Q_{i\oplus3}).
  \]
  The first parenthesis is a difference of projections and hence has norm at most one. Right multiplication by $G$ adds at most $2\delta$. The unrounded bound for a single pair therefore gives $(214+2+214)\delta=430\delta$, proving both claims.
\end{proof}

Recall $T_i^A=(I,f,g,-fg)_i$ and $T_i^B=(I,F,G,-FG)_i$ for $0\leq i<4$. Alice's $T_i^A$ approximates $(T_i^B)^{\dagger}$ with error at most $4\delta$. Cross-party commutation therefore gives
\[
  \left\lVert T_i^AT_i^BQ_i\Omega -T_i^BQ_i(T_i^B)^{\dagger}\Omega\right\rVert\leq4\delta.
\]
In particular, for the extracted branch vectors $\eta_i=T_i^AT_i^BQ_i\Omega$,
\begin{equation}
  \eta_0=Q_0\Omega,\qquad \lVert\eta_i-\eta_0\rVert\leq804\delta\leq10\,000\delta \quad(1\leq i<4).
  \label{eq:rt:high-branches}
\end{equation}

\subsection{Evaluating the coordinate changes for the fifth branch}
\label{app:coordinate-evaluation}

The high-branch estimate is now established. To compare the fifth branch with the first, we must express Bob's two pair projections in a common coordinate system. Define $v_i=T_i^BQ_0\Omega$ in the original purified device space. These four vectors need not be orthogonal; the comparison matrices below serve only to track operator actions on them. The commuting reflections act by the exact rules
\begin{equation}
  \begin{array}{c|rrrr}
    & v_0 & v_1 & v_2 & v_3 \\ \hline
    F & v_1 & v_0 & -v_3 & -v_2 \\
    G & v_2 & -v_3 & v_0 & -v_1
  \end{array}
  \label{eq:rt:high-four-vectors}
\end{equation}
By Lemma~\ref{lem:rt:high-endpoints}, the difference of projections $T_i^BQ_0(T_i^B)^{\dagger}-Q_i$ has norm on $\Omega$ at most $800\delta$. The right factor $T_i^BQ_0$ has an Alice contraction approximation with error at most $6\delta$. Hence
\begin{equation}
  \lVert(I-Q_i)v_i\rVert\leq806\delta\leq\kappa\delta, \qquad \kappa=10\,000.
  \label{eq:rt:high-support}
\end{equation}
For any diagonal phase unitary $\Delta=\sum_{i=0}^4d_iQ_i$, this implies
\begin{equation}
  \lVert\Delta v_i-d_iv_i\rVert\leq2\kappa\delta.
  \label{eq:rt:high-phases}
\end{equation}
The operator $\Delta-d_iI$ vanishes on the range of $Q_i$ and has norm at most two.

After right multiplication by $T_i^BQ_0$, either difference in \eqref{eq:rt:high-signed-reflections} has error at most $(6+2\cdot6)\delta=18\delta$. Apply~\eqref{eq:rt:high-phases} to $J$ or $K$, after the exact actions in~\eqref{eq:rt:high-four-vectors}. The resulting approximate actions are
\begin{equation}
  \begin{array}{c|rrrr}
    & v_0 & v_1 & v_2 & v_3 \\ \hline
    \widehat F & v_1 & v_0 & v_3 & v_2 \\
    \widehat G & v_2 & v_3 & v_0 & v_1
  \end{array}
  \qquad \text{with each error at most }50\,000\delta.
  \label{eq:rt:high-other-actions}
\end{equation}
Here the actual bound is $(2\kappa+18)\delta$.

To define the four-dimensional comparison matrices, replace each $v_i$ in the displayed actions by the standard basis vector $e_i$ of $\mathbb C^4$. On this comparison space, $Q_i$ is replaced by $e_ie_i^{\dagger}$ for $0\le i<4$, while $Q_4$ is replaced by zero. Put
\[
  h=\tfrac12(-1,1,-1,-1),\qquad k=\tfrac12(1,1,-1,-1).
\]
The Bob projections and their four-dimensional comparison matrices are
\begin{equation}
  \begin{array}{c|c|c}
    \text{projection} & \text{exact operator formula} & \text{comparison matrix} \\ \hline
    C_1=B_{12}+B_{13} & (I-F)(I+G)/4 & hh^{\dagger} \\
    C_2=B_{22}+B_{23} & (I+\widehat F)(I-\widehat G)/4 & kk^{\dagger} \\
  \end{array}
  \label{eq:rt:high-pair-projections}
\end{equation}
For $y=1$, the comparison action on the four vectors is exact. For $y=2$, expand the product and telescope its quadratic term using \eqref{eq:rt:high-other-actions}. Each linear term contributes at most $50\,000\delta$ and the quadratic term at most $100\,000\delta$; the factor $1/4$ gives a total error of $50\,000\delta$ per column. The signed permutations above determine the comparison matrices because each sign vector spans the simultaneous eigenspace for the prescribed pair of eigenvalues.

We now evaluate the coordinate changes $V_v$ defined in~\eqref{eq:rt:charts}. For $v=h$ or $k$, the comparison matrices satisfy
\[
  E(F)E(G)e_0=\tfrac12(e_0+ie_1+ie_2+e_3), \qquad D_vE(F)E(G)e_0=v,
\]
and hence $V_vv=e_0$ in the comparison space. The physical operator $V_v$ is unitary for every strategy. In the product $V_v=E(G)^{\dagger}E(F)^{\dagger}D_v^{\dagger}$, the diagonal factor acts first. Its error is bounded by \eqref{eq:rt:high-phases}, while the subsequent actions of $E(F),E(G)$ are exact. Thus the action of $V_v$ on each $v_i$ differs from the comparison action by at most $2\kappa\delta=20\,000\delta$.

The absolute values in each column of a rank-one projection onto one of these two sign vectors sum to one. Combining the preceding estimates therefore gives
\begin{equation}
  \left\lVert V_vC_yv_i-c_iv_0\right\rVert \leq
  \begin{cases}
    20\,000\delta, & y=1, \\
    70\,000\delta, & y=2,
  \end{cases}
  \label{eq:rt:high-coordinate-action}
\end{equation}
where $v$ is the sign vector assigned to $C_y$ in \eqref{eq:rt:high-pair-projections}, and $c_i\in\{1/2,-1/2\}$ is defined by $v=(c_0,c_1,c_2,c_3)$. The comparison term is therefore a scalar multiple of $v_0=Q_0\Omega$.

\subsection{Putting the four conditional vectors in common coordinates}
\label{app:conditional-vectors}

For $x=2,3$ and $y=1,2$, put
\begin{equation}
  U_2=I,\qquad U_3=g,\qquad r_2=0,\qquad r_3=2,
  \qquad V_1=V_h,\qquad V_2=V_k.
  \label{eq:rt:coordinates}
\end{equation}
Write $\eta_0=Q_0\Omega$ and define
\begin{equation}
  E_x=A_{x2}+A_{x3},\qquad C_y=B_{y2}+B_{y3},
  \qquad w_{xy}=U_xV_yE_xC_y\Omega.
  \label{eq:rt:low-pairs}
\end{equation}
The definitions of $C_1,C_2$ here agree with the pair projections in~\eqref{eq:rt:high-pair-projections}.
The vectors $w_{xy}$ are the unnormalized components for answers $\{2,3\}$ on both sides, after the local coordinate changes. We estimate their high and low parts separately. The signs that will enter the binary tests are
\begin{equation}
  s_{21}=-1,\qquad s_{31}=-1,\qquad s_{22}=+1,\qquad s_{32}=-1.
  \label{eq:rt:signs}
\end{equation}

\begin{lemma}[High parts of the conditional vectors]
  \label{lem:rt:conditional}
  For $x=2,3$ and $y=1,2$,
  \begin{equation}
    \left\lVert U_xV_yC_yQ_{r_x}\Omega
        -\frac{s_{xy}}2Q_0\Omega\right\rVert\leq\tau,
    \qquad \tau=100\,000\delta.
    \label{eq:rt:high-targets}
  \end{equation}
\end{lemma}
\begin{proof}
  The Bob approximations to $U_2,U_3$ are $I,G$, with errors zero and $2\delta$. Commute $U_x$ to the right and replace it by that approximation. The resulting vector is $V_yC_yQ_{r_x}T_{r_x}^B\Omega$. Since each $T_r^B$ is self-adjoint, Lemma~\ref{lem:rt:high-endpoints} gives
  \[
    \lVert(Q_rT_r^B-T_r^BQ_0)\Omega\rVert
    =\lVert(T_r^BQ_rT_r^B-Q_0)\Omega\rVert\leq800\delta.
  \]
  Thus the vector differs from $V_yC_yv_{r_x}$ by at most $802\delta$. Equation~\eqref{eq:rt:high-coordinate-action} bounds the remaining error by $70\,000\delta$. The coordinate $r_x$ of $h$ or $k$ is $s_{xy}/2$, proving the assertion.
\end{proof}

\paragraph{Low components.}
The coarse agreements give
\begin{gather}
  \lVert(A_{04}-Q_4)\Omega\rVert\leq\delta,
  \qquad \lVert(A_{10}-Q_4)\Omega\rVert\leq\delta,
  \label{eq:rt:low-singleton-sync}\\
  \lVert(E_x-Q_{r_x}-Q_4)\Omega\rVert\leq\delta.
  \label{eq:rt:low-pair-sync}
\end{gather}
The exact identities $fA_{04}=-A_{04}$ and $gA_{10}=A_{10}$ imply
\begin{equation}
  \lVert(F+I)Q_4\Omega\rVert\leq4\delta,
  \qquad \lVert(G-I)Q_4\Omega\rVert\leq4\delta.
  \label{eq:rt:low-scalar-restrictions}
\end{equation}
For the first inequality, replacing $Q_4\Omega$ by $A_{04}\Omega$ costs at most $2\delta$, and
\[
  (F+I)A_{04}\Omega=A_{04}(F-f)\Omega
\]
costs at most another $2\delta$. The proof of the second inequality is the same, using $A_{10}$.

For $E(L)=(I+\mathrm iL)/\sqrt2$ from~\eqref{eq:rt:charts}, a two-term telescoping estimate now gives
\[
  \lVert(E(F)E(G)-I)Q_4\Omega\rVert\leq4\sqrt2\delta.
\]
Indeed, on this vector each factor differs by at most $2\sqrt2\delta$ from its scalar value $(1-\mathrm i)/\sqrt2$ or $(1+\mathrm i)/\sqrt2$, whose product is one. Since $D_vQ_4=Q_4$, multiplication by the inverse of $D_vE(F)E(G)$ gives
\begin{equation}
  \lVert(V_y-I)Q_4\Omega\rVert\leq6\delta.
  \label{eq:rt:low-bob-coordinate-action}
\end{equation}
On Alice's side $U_2=I$, and
\[
  \lVert(g-I)Q_4\Omega\rVert
  =\lVert(g-I)(Q_4-A_{10})\Omega\rVert\leq2\delta.
\]
Consequently,
\begin{equation}
  \lVert(U_xV_y-I)Q_4\Omega\rVert\leq8\delta.
  \label{eq:rt:low-joint-coordinate-action}
\end{equation}
When Alice answers $4$ at question zero, all Bob answers outside $\{2,3\}$ at questions one and two are forbidden. Thus $\lVert A_{04}(I-C_y)\Omega\rVert\leq\delta$. Replacing $A_{04}\Omega$ by $Q_4\Omega$ gives
\begin{equation}
  \lVert(I-C_y)Q_4\Omega\rVert\leq2\delta.
  \label{eq:rt:low-bob-pair-leakage}
\end{equation}

In $w_{xy}$, commute $E_x$ past $C_y$ and replace $E_x\Omega$ by $(Q_{r_x}+Q_4)\Omega$ at cost $\delta$. Lemma~\ref{lem:rt:conditional} controls the high part. Removing $C_y$ from the low part costs $2\delta$, and the local coordinate changes cost $8\delta$. We have therefore proved
\begin{equation}
  \left\lVert w_{xy}-\left(\frac{s_{xy}}2\eta_0+Q_4\Omega\right)\right\rVert
  \leq\tau+11\delta.
  \label{eq:rt:low-conditional}
\end{equation}

\subsection{Using the four binary tests to align the fifth branch}
\label{app:fifth-branch-alignment}

The conditional vectors now have the required common form. We first isolate the consequence of the three agreement tests and one disagreement test, then apply it with the errors just obtained. The following estimate uses a local reflection that approximately reverses the sign of the high component while preserving the low component. It does not divide by either component's weight. A reflection means a self-adjoint unitary.

\begin{lemma}
\label{lem:robust-two-subspaces}
Let $h,\phi$ be vectors, set $z_\pm=\pm h+\phi$, and let
$J_A,R,R'$ be Alice reflections and $C,D'$ Bob reflections.
Suppose that $\norm{J_Az_--z_+}\le b$ and each of
\begin{equation}
 (R-C)z_-,\qquad(R'-C)z_-,\qquad(R+D')z_+,\qquad(R'-D')z_-
 \label{eq:rs:binary-quartet}
\end{equation}
has norm at most $a$. Then $\norm{RC\phi+h}\le2a+b$.
\end{lemma}
\begin{proof}
Subtracting the second and fourth residuals gives
$\norm{(C-D')z_-}\le2a$. Since $J_A$ commutes with Bob's reflections,
\[
 \norm{(C-D')z_+}
 \le\norm{J_A(C-D')z_-}+\norm{(C-D')(z_+-J_Az_-)}
 \le2a+2b.
\]
Thus $e_-=(R-C)z_-$ and $e_+=(R+C)z_+$ satisfy
$\norm{e_-}\le a$ and $\norm{e_+}\le3a+2b$.
The identity $2(RC\phi+h)=R(e_+-e_-)$, using $R^2=I$,
proves the claim.
\end{proof}

For two-outcome restrictions of projective measurements, write $E=A_a+A_b$, $C_0=B_c+B_d$, $r=I-2A_a$, and $s=I-2B_c$. Orthogonality of the four joint outcomes gives
\begin{align}
  \lVert(r-s)EC_0\Omega\rVert^2
    &=4\bigl(\lVert A_aB_d\Omega\rVert^2+\lVert A_bB_c\Omega\rVert^2\bigr),
    \label{eq:rt:low-binary-minus}\\
  \lVert(r+s)EC_0\Omega\rVert^2
    &=4\bigl(\lVert A_aB_c\Omega\rVert^2+\lVert A_bB_d\Omega\rVert^2\bigr).
    \label{eq:rt:low-binary-plus}
\end{align}
For example, $r-s$ has eigenvalues $-2,2$ on the two unequal outcomes and vanishes on the other two. Both identities persist after the corresponding local unitary changes of coordinates.

Define reflections
\begin{align*}
  \mathsf R&=U_2(I-2A_{22})U_2^{\dagger},
  &\mathsf R'&=U_3(I-2A_{32})U_3^{\dagger},\\
  \mathsf C&=V_1(I-2B_{12})V_1^{\dagger},
  &\mathsf D&=V_2(I-2B_{22})V_2^{\dagger}.
\end{align*}
Each row of Table~\ref{tab:low} has forbidden probability at most $\delta^2$. Equations~\eqref{eq:rt:low-binary-minus}--\eqref{eq:rt:low-binary-plus} therefore bound each vector
\begin{equation}
  (\mathsf R-\mathsf C)w_{21},\quad
  (\mathsf R'-\mathsf C)w_{31},\quad
  (\mathsf R+\mathsf D)w_{22},\quad
  (\mathsf R'-\mathsf D)w_{32}
  \label{eq:rt:low-actual}
\end{equation}
by $2\delta$.

Apply the lemma directly with
\[
 h=\tfrac12\eta_0,\qquad\phi=Q_4\Omega,\qquad
 z_\pm=\pm\tfrac12\eta_0+Q_4\Omega,\qquad J_A=I-2U.
\]
The agreement $\norm{(U-Q_0-Q_1)\Omega}\le\delta$ gives
\[
 \norm{(I-U)Q_0\Omega}\le\delta,\qquad
 \norm{UQ_4\Omega}\le\delta,
 \qquad\norm{J_Az_--z_+}\le3\delta.
\]
Indeed, the first two bounds follow by multiplying the agreement
by $Q_0$ and $Q_4$, respectively, and
$J_Az_--z_+=-(I-U)Q_0\Omega-2UQ_4\Omega$.
By~\eqref{eq:rt:low-conditional}, replacing each $w_{xy}$
in~\eqref{eq:rt:low-actual} by its corresponding $z_\pm$ costs
at most $2(\tau+11\delta)$. Thus the lemma applies with
$a=2\tau+24\delta$ and $b=3\delta$.
Since $T_4^A=\mathsf R$ and $T_4^B=-\mathsf C$, it yields
\begin{equation}
 \boxed{\left\lVert T_4^AT_4^BQ_4\Omega-\frac12\eta_0\right\rVert
 \le4\tau+51\delta\le400\,200\delta.}
 \label{eq:rt:low-final}
\end{equation}
No normalization or division by a component weight is used.

\Needspace{10\baselineskip}
\section{Exact identification of the strategy}
\label{sec:proof}

We prove Theorem~\ref{thm:strategy-exact} directly for arbitrary POVMs and mixed states. The proof has four steps. Coarse agreements give five matching subspaces; relations within the first four identify their state components and measurement effects; the remaining binary tests fix the fifth component and its relative amplitude; and positivity extends the resulting classification from the reduced-state supports to the original effects. This argument constructs canonical local coordinates. Appendix~\ref{app:robust-measurements} then shows that the fixed circuit of Section~\ref{sec:robust-state} realizes these coordinates and uses that fact to prove robustness.

\subsection{Agreement relations and five matching subspaces}
\label{app:exact-supports}

Fix a perfect finite-dimensional strategy and a purification $\Omega\in\mathcal H_A\otimes\mathcal H_B\otimes\mathcal H_E$ of its state. Compress each local effect to the support of the corresponding reduced state, and write the compressed effects as $A_{xa}$ and $B_{yb}$. This preserves the probabilities and the POVM sums. Throughout this appendix, the symbols $A_{xa},B_{yb}$ refer to these compressed effects until we return to the original POVMs in Appendix~\ref{app:exact-effects}. On the compressed spaces the reduced states are faithful: for any local operator $L$, the equality $L\Omega=0$ implies $L=0$.

Equation~\eqref{eq:exact-agreement} shows that exact agreement of opposite-party effects forces them to be projections on these supports. We call the resulting identity $E\Omega=F\Omega$ an \emph{agreement relation}.
Products transfer across an agreement relation in reverse order. For example, if $E_i\Omega=F_i\Omega$ for $i=1,2$, then $E_1E_2\Omega=F_2F_1\Omega$. Faithfulness promotes the resulting vector identities to local operator identities. We also use the following consequence of positivity. If two coarse sums of one POVM are projections, they commute, and their product is the sum over the intersection of the two outcome sets. Indeed, every effect is supported in each coarse projection containing it and is orthogonal to each coarse projection whose outcome set excludes it.

Apply these facts to Table~\ref{tab:coarse}, using the coarse operators $U,W,V,T$ defined in~\eqref{eq:coarse-four}. They are projections agreeing with Bob's coarse sums over $\{0,1\}$, $\{2,3\}$, $\{0,2\}$, and $\{1,3\}$ at question zero. Intersecting those sums shows that all five $Q_i=B_{0i}$ are projections. Their Alice counterparts are
\begin{equation}
  \begin{gathered}
    P_0=UV,\qquad P_1=UT,\qquad P_2=WV,\qquad P_3=WT,
    \qquad P_4=A_{04}=A_{10},\\
    P_i\Omega=Q_i\Omega,\qquad
    P_iP_j=\delta_{ij}P_i,\qquad \sum_{i=0}^{4}P_i=I.
  \end{gathered}
  \label{eq:five-sectors}
\end{equation}
The factors in each product commute, by their agreement with commuting Bob projections. The other rows involving Bob's question zero determine the following pair sums:
\begin{equation}
  \begin{array}{c|ccc}
    x&\text{first group}&\text{second group}&\text{third group}\\ \hline
    0&\{0,1\}:P_0+P_1&\{2,3\}:P_2+P_3&\{4\}:P_4\\
    1&\{1,2\}:P_0+P_2&\{3,4\}:P_1+P_3&\{0\}:P_4\\
    2&\{0,1\}:P_1+P_2&\{2,3\}:P_0+P_4&\{4\}:P_3\\
    3&\{0,1\}:P_0+P_3&\{2,3\}:P_2+P_4&\{4\}:P_1
  \end{array}
  \label{eq:pair-supports}
\end{equation}
Here an entry $S:P$ means $\sum_{a\in S}A_{xa}=P$.

\subsection{Four equal components and their measurement effects}
\label{app:exact-high}

The pair supports are now fixed. We next determine the differences within the pairs on the first four subspaces. Their off-diagonal blocks will supply the local coordinate identifications.

\begin{lemma}
  \label{lem:four-components}
  There are local identifications, on the reduced-state supports,
  \[
    \mathcal H_X=(\C^4\otimes\mathcal K_X)\oplus\mathcal L_X,
    \qquad X=A,B,
  \]
  under which $P_i,Q_i$ are the coordinate projections and
  \begin{equation}
    \Omega=\sum_{i=0}^{3}\ket{ii}\eta+\phi,
    \qquad
    \eta\in\mathcal K_A\otimes\mathcal K_B\otimes\mathcal H_E,
    \quad
    \phi\in\mathcal L_A\otimes\mathcal L_B\otimes\mathcal H_E.
    \label{eq:four-components}
  \end{equation}
  Alice's effects at questions zero and one that are supported in the first four components are the corresponding projections in Table~\ref{tab:alice}, tensored with the residual identity. Their remaining effects are the projection onto $\mathcal L_A$. The same high-component identification holds for outcomes zero and one at questions two and three.
\end{lemma}

\begin{proof}
  We first obtain commuting coarse projections, then use their block relations to identify the four subspaces. The last four agreements in Table~\ref{tab:coarse} give the projections
  \begin{equation}
    \begin{aligned}
      F_1&=A_{01}+A_{02},&G_1&=A_{10}+A_{12}+A_{13},\\
      F_2&=A_{00}+A_{02},&G_2&=A_{12}+A_{14}.
    \end{aligned}
    \label{eq:four-coarse-projections}
  \end{equation}
  Their Bob counterparts are, in the same order,
  \[
    B_{10}+B_{11},\quad B_{10}+B_{12}+B_{13},\quad
    B_{20}+B_{24},\quad B_{20}+B_{21}.
  \]
  The two counterparts at each question commute. Transferring their products gives $[F_1,G_1]=[F_2,G_2]=0$.

  Restricting $F_1$ to the first two pair supports in~\eqref{eq:pair-supports} shows that $A_{01}$ and $A_{02}$ are projections. Restricting $G_1$ to $P_0+P_2$ and $P_1+P_3$ gives the same conclusion for $A_{12}$ and $A_{13}$. Their complements within each pair are projections as well. These restrictions are products of commuting projections with agreement relations, so every effect at questions zero and one has an agreement relation.

  Set $H=P_0+P_1+P_2+P_3$. Define the four differences
  \[
    r_{01}=A_{01}-A_{00},\qquad r_{23}=A_{03}-A_{02},\qquad
    s_{02}=A_{12}-A_{11},\qquad s_{13}=A_{14}-A_{13}.
  \]
  Each squares to the projection onto its indicated pair. On $H\mathcal H_A$, put
  \begin{equation}
    r=r_{01}-r_{23},\qquad s=s_{02}-s_{13},\qquad
    r'=r_{01}+r_{23},\qquad s'=s_{02}+s_{13}.
    \label{eq:high-reflections}
  \end{equation}
  These are reflections. Since, on this subspace,
  \[
    r=2F_1-H,\qquad s=2G_1-H,\qquad
    r'=H-2F_2,\qquad s'=2G_2-H,
  \]
  we have $[r,s]=[r',s']=0$.

  To use the high-component constraints, define
  \[
    C_1=F_1(I-G_1)+(I-F_1)G_1,\qquad
    C_2=F_2G_2+(I-F_2)(I-G_2).
  \]
  These projections agree with $B_{11}+B_{12}+B_{13}$ and $B_{20}+B_{22}+B_{23}$, respectively. On $H$, they equal $(H-rs)/2$ and $(H-r's')/2$.

  If an effect $E$ has zero probability with the complement of a Bob projection agreeing with a projection $C$, then
  $\sqrt E(I-C)\Omega=0$. Faithfulness gives $E(I-C)=(I-C)E=0$, so $E$ is supported in $C$. If its complementary effect within a pair is supported in $I-C$, their sum commutes with $C$, and $E$ is their sum multiplied by $C$. Applying this observation to both outcomes in each indicated pair in Table~\ref{tab:high} gives
  \begin{equation}
    \begin{aligned}
      A_{20}&=(P_1+P_2)C_1=(P_1+P_2)C_2,\\
      A_{30}&=(P_0+P_3)(I-C_1)=(P_0+P_3)C_2,\\
      (P_1+P_2)(rs-r's')&=0,\qquad
      (P_0+P_3)(rs+r's')=0.
    \end{aligned}
    \label{eq:high-product-identities}
  \end{equation}
  In the first two lines, every product is a product of commuting projections.

  The product identities give the coordinate exchanges directly. On $H\mathcal H_A$, define
  \[
    Z_1=P_0+P_1-P_2-P_3,\qquad Z_2=P_0-P_1+P_2-P_3.
  \]
  The pair supports imply $[r,Z_1]=[s,Z_2]=0$ and $r'=Z_1r$, $s'=Z_2s$. Combining the last two identities in~\eqref{eq:high-product-identities} gives
  \[
    r's'=-Z_1Z_2rs.
  \]
  Substitute $r'=Z_1r$ and $s'=Z_2s$, then cancel $Z_1$ on the left and $s$ on the right. This yields $rZ_2=-Z_2r$. Since $r's'=s'r'$ and $rs=sr$, the same equation reads $Z_2sZ_1r=-Z_2Z_1sr$, giving $sZ_1=-Z_1s$. Thus
  \begin{equation}
    \{r,Z_2\}=0,\qquad \{s,Z_1\}=0.
    \label{eq:high-anticommutation}
  \end{equation}
  Hence $r'$ exchanges $P_0\leftrightarrow P_1$ and $P_2\leftrightarrow P_3$, while $s'$ exchanges $P_0\leftrightarrow P_2$ and $P_1\leftrightarrow P_3$. Set $\mathcal K_A=P_0\mathcal H_A$ and use $r',s',r's'$ to identify coordinates one, two, and three with coordinate zero. Their commutation makes all four pair differences positive swaps with identity off-diagonal blocks. The first two lines of~\eqref{eq:high-product-identities} then give the projections onto $e_1-e_2$ and $e_0-e_3$, tensored with $I_{\mathcal K_A}$. Their complements are the projections onto the corresponding positive sums. This determines Alice's effects in the statement.

  It remains to transfer these coordinates to Bob and identify the state. The four coordinate projections and these swaps generate the standard matrix units $E^A_{ij}$ on $\C^4\otimes\mathcal K_A$. Each generator has a self-adjoint Bob counterpart agreeing on $\Omega$. Transfer each product in reverse order to define $E^B_{ji}$ with
  \begin{equation}
    E^A_{ij}\Omega=E^B_{ji}\Omega.
    \label{eq:matrix-unit-agreement}
  \end{equation}
  This definition is independent of the polynomial chosen: a vanishing polynomial on Alice transfers to an operator annihilating $\Omega$ on Bob, which vanishes by faithfulness. The same argument transfers multiplication and adjoints, so the $E^B_{ij}$ are matrix units with diagonal projections $Q_i$. They identify Bob's first four subspaces with $\C^4\otimes\mathcal K_B$.

  The diagonal agreements in~\eqref{eq:five-sectors} give
  $\Omega=\sum_{i=0}^{3}\ket{ii}\eta_i+\phi$.
  For each $i,j$, equation~\eqref{eq:matrix-unit-agreement} equates the coefficients of $\ket{ij}$ and gives $\eta_i=\eta_j$. Thus all four vectors equal a common $\eta$, proving~\eqref{eq:four-components}.
\end{proof}

The same coarse projections determine Bob's effects except for one binary splitting at each of his last two questions. Their joint eigenspaces are
\begin{equation}
  \begin{array}{c|cccc}
    y&\{0\}&\{1\}&\{2,3\}&\{4\}\\ \hline
    1&c_0\mathcal K_B&c_1\mathcal K_B&h\mathcal K_B\oplus\mathcal L_B&h'\mathcal K_B\\
    2&d_0\mathcal K_B&d_1\mathcal K_B&k\mathcal K_B\oplus\mathcal L_B&k'\mathcal K_B
  \end{array}
  \label{eq:bob-pair-supports}
\end{equation}
Here the vectors are those in Table~\ref{tab:bob}, and $v\mathcal K_B$ denotes $\{v\otimes\xi:\xi\in\mathcal K_B\}$. To verify the table, let $\overline F_i,\overline G_i$ be the Bob counterparts of~\eqref{eq:four-coarse-projections}. At question one the four projections are
\[
  \overline F_1\overline G_1,\quad
  \overline F_1(I-\overline G_1),\quad
  (I-\overline F_1)\overline G_1,\quad
  (I-\overline F_1)(I-\overline G_1).
\]
At question two they are
\[
  \overline F_2\overline G_2,\quad
  (I-\overline F_2)\overline G_2,\quad
  (I-\overline F_2)(I-\overline G_2),\quad
  \overline F_2(I-\overline G_2).
\]
Substitution of the positive swaps gives the high-coordinate lines displayed in~\eqref{eq:bob-pair-supports}. On the low subspace, the eigenvalues of $(\overline F_1,\overline G_1)$ are $(0,1)$, and those of $(\overline F_2,\overline G_2)$ are $(0,0)$. Thus both binary pairs contain the entire low subspace and the singleton effects contain none of it.

\subsection{The fifth component and the remaining binary differences}
\label{app:exact-low}

Only the differences within the pairs $\{2,3\}$ at Alice's questions two and three and Bob's questions one and two remain undetermined. The next lemma identifies all four differences and the fifth state component at once.

\begin{lemma}
  \label{lem:fifth-component}
  There are unitaries $U_A:\mathcal L_A\to\mathcal K_A$ and $U_B:\mathcal L_B\to\mathcal K_B$ such that
  \begin{equation}
    (U_A\otimes(-U_B)\otimes I_E)\phi=\eta/2.
    \label{eq:fifth-amplitude}
  \end{equation}
  In the pair coordinates specified below, Alice's two remaining binary differences are both
  $\left(\begin{smallmatrix}0&U_A\\U_A^{\dagger}&0\end{smallmatrix}\right)$,
  and Bob's are both
  $\left(\begin{smallmatrix}0&U_B\\U_B^{\dagger}&0\end{smallmatrix}\right)$.
\end{lemma}

\begin{proof}
  By~\eqref{eq:pair-supports} and~\eqref{eq:bob-pair-supports}, Alice's pairs $\{2,3\}$ act on $e_0\mathcal K_A\oplus\mathcal L_A$ and $e_2\mathcal K_A\oplus\mathcal L_A$ at questions two and three. Bob's corresponding pairs act on $h\mathcal K_B\oplus\mathcal L_B$ and $k\mathcal K_B\oplus\mathcal L_B$. Identify each high component with its residual space, keeping the low space fixed. Write the differences, second effect minus first, as
  \[
    R=A_{23}-A_{22},\qquad R'=A_{33}-A_{32},\qquad
    C=B_{13}-B_{12},\qquad D'=B_{23}-B_{22}
  \]
  in these pair coordinates.

  Projecting~\eqref{eq:four-components} onto each pair of supports gives an unnormalized conditional vector on
  $(\mathcal K_A\oplus\mathcal L_A)\otimes
  (\mathcal K_B\oplus\mathcal L_B)\otimes\mathcal H_E$.
  Since $\langle e_0,h\rangle=\langle e_2,h\rangle=\langle e_2,k\rangle=-1/2$ and $\langle e_0,k\rangle=1/2$, the vectors at question pairs $(2,1),(3,1),(2,2),(3,2)$ are, respectively,
  \begin{equation}
    z_-,\quad z_-,\quad z_+,\quad z_-,
    \qquad z_\pm=\pm\eta/2+\phi.
    \label{eq:conditional-signs}
  \end{equation}
  Both $z_+$ and $z_-$ have full local support on these pair spaces. Indeed, their marginals are direct sums of one quarter of the corresponding marginal of $\eta$ and the marginal of $\phi$. Those component marginals are faithful on $\mathcal K_X$ and $\mathcal L_X$, because the original reduced states are faithful.

  Table~\ref{tab:low} forbids unequal binary outcomes at the first, second, and fourth question pairs, and equal outcomes at the third. Applying~\eqref{eq:exact-agreement} to the restricted binary POVMs makes all four differences reflections and gives exactly~\eqref{eq:four-binary-relations}. We have now verified the support, sign, and faithfulness hypotheses used in Section~\ref{sec:measurement-exact}. The calculation~\eqref{eq:measurement-anticommutation} forces the block forms~\eqref{eq:measurement-binary-blocks}, and~\eqref{eq:measurement-amplitude} proves~\eqref{eq:fifth-amplitude}.
\end{proof}

The unitary identifications between the high and low spaces also rule out vanishing summands, so no positive outcome probability was assumed.

\subsection{The canonical strategy and the original POVM effects}
\label{app:exact-effects}

We now assemble the coordinate identifications and recover the individual effects from their pair sums and differences. We then show that the original effects preserve the occupied supports.

Define local unitaries $\mathcal C_X$ from the reduced-state supports to $\C^5\otimes\mathcal K_X$. Keep the first four coordinates and map the low subspaces by
\[
  \xi\in\mathcal L_A\longmapsto\ket4\otimes U_A\xi,
  \qquad
  \zeta\in\mathcal L_B\longmapsto\ket4\otimes(-U_B)\zeta.
\]
Equations~\eqref{eq:four-components} and~\eqref{eq:fifth-amplitude} give
\begin{equation}
  (\mathcal C_A\otimes\mathcal C_B\otimes I_E)\Omega
   =\left(\sum_{i=0}^{3}\ket{ii}+\tfrac12\ket{44}\right)\eta
   =\ket{\psistar}\otimes\chi,
  \qquad \chi=\frac{\sqrt{17}}2\eta.
  \label{eq:canonical-state}
\end{equation}
Normalization gives $\norm{\eta}^{2}=4/17$, so $\chi$ is normalized.

The same identifications determine every labeled effect. Lemma~\ref{lem:four-components} gives Alice's first two rows and her high pairs at questions two and three. Lemma~\ref{lem:fifth-component} makes her two low-pair differences positive swaps in the new coordinates, while~\eqref{eq:pair-supports} fixes their sums. Their ordered effects are therefore the projections onto $e_0-e_4,e_0+e_4$ and $e_2-e_4,e_2+e_4$, respectively. The last effects are the singleton projections onto $e_3$ and $e_1$.

On Bob's side,~\eqref{eq:bob-pair-supports} fixes the singleton effects and both pair sums. The minus sign in his low-space identification makes the binary differences negative swaps. If the high line is $v$, a negative swap gives the ordered rank-one projections onto $v+e_4$ and $v-e_4$. Thus these are precisely the effects in Table~\ref{tab:bob}. The equalities $R=R'$ and $C=D'$ ensure that one low-space identification works at both questions on each side. We have proved, with all labels unchanged,
\begin{equation}
  \mathcal C_A A_{xa}\mathcal C_A^{\dagger}
    =A^\star_{xa}\otimes I_{\mathcal K_A},\qquad
  \mathcal C_B B_{yb}\mathcal C_B^{\dagger}
    =B^\star_{yb}\otimes I_{\mathcal K_B}.
  \label{eq:canonical-effects}
\end{equation}

It remains to pass from compressed effects back to the original POVMs. Every compressed effect in~\eqref{eq:canonical-effects} is a projection. If $0\leq E\leq I$, $s$ is a support projection, and $sEs$ is a projection on $s\mathcal H$, then
\begin{equation}
  ((I-s)Es)^{\dagger}((I-s)Es)
    =sE^2s-(sEs)^2
    \leq sEs-(sEs)^2=0.
  \label{eq:no-leakage}
\end{equation}
Thus $E$ preserves the support, and its compression is its restriction. Enlarge the residual spaces if necessary and extend $\mathcal C_A,\mathcal C_B$ to isometries $W_A,W_B$ on the original local Hilbert spaces. Equations~\eqref{eq:canonical-state}--\eqref{eq:no-leakage} give the state and measurement-action identities of Theorem~\ref{thm:strategy-exact}, including products of the effects. The local identifications depend only on the strategy and preserve the purification system. Tracing out that system yields the mixed-state factorization with the residual density operator $\sigma=\Tr_E\ket\chi\bra\chi$. This completes the proof.

\Needspace{10\baselineskip}
\section{Simultaneous robustness of the measurement actions}
\label{app:robust-measurements}

We supply the exact identities needed for the compactness argument in Section~\ref{sec:measurement-compactness}. First we verify that the fixed extraction circuit of~\eqref{eq:rs:isometries} implements the classification in Appendix~\ref{sec:proof}, then extend its zero-loss identities to all commuting representations. The final subsection lifts the resulting projective-strategy self-test to mixed states and general POVMs, completing Theorem~\ref{thm:strategy-robust}.

Throughout this appendix, $\Psi_\star=\ket{\psistar}\!\bra{\psistar}$. Until the final subsection, measurement symbols denote projective measurements on the full device spaces, rather than the compressed POVMs used at the start of Appendix~\ref{sec:proof}.

\subsection{The fixed circuit at perfect score}
\label{app:rm:fixed}

The canonical unitaries $\mathcal C_X$ in Appendix~\ref{app:exact-effects} were constructed only on the occupied supports. The fixed maps $\mathcal V_A,\mathcal V_B$ are defined on the full device spaces at every score, with noncommutative polynomial coefficients. We first verify that they realize the same extracted coordinates at perfect score, with the coordinate-zero spaces and Alice's fixed flag included in the residual systems. This fixes the extraction prescription before passing to limits.

\begin{lemma}\label{lem:rm:fixed-circuit}
  For a perfect projective strategy, the fixed isometries satisfy, on the respective local reduced-state supports,
  \begin{equation}
    \mathcal V_A A_{xa}=(A^\star_{xa}\otimes I)\mathcal V_A,
    \qquad
    \mathcal V_B B_{yb}=(B^\star_{yb}\otimes I)\mathcal V_B.
    \label{eq:rm:fixed-intertwining}
  \end{equation}
  Their joint action on the state extracts $\ket{\psistar}$.
\end{lemma}

\begin{proof}
  The classification in Appendix~\ref{sec:proof} makes every local support invariant under the measurement projections. At perfect score, Alice's label operators obey $K_i=P_i$ and her bad-flag operators vanish on that support. In the coordinates of Lemma~\ref{lem:four-components}, the first four transport products $T_i^AP_i$ and $T_i^BQ_i$ send coordinate $i$ to coordinate zero with positive identity block. For $i=3$, the minus sign in $-fg$ and $-FG$ cancels the negative block of the product of the two reflections.

  Alice's last transport $I-2A_{22}$ has block $U_A$ from the low space to coordinate zero. On Bob's support, the chart $V_h$ sends $h$ to $e_0$ and fixes the low coordinate: the two low-coordinate phases of $E(F)E(G)$ multiply to one, as does the low entry of $D_h$. Consequently $-V_h(I-2B_{12})V_h^{\dagger}$ has low-to-zero block $-U_B$. These are exactly the two low-space identifications used in~\eqref{eq:canonical-state}.

  On the respective supports, the coefficients $v_i^A=T_i^AP_i$ and $v_i^B=T_i^BQ_i$ therefore implement the canonical coordinate identifications. Their products $e_{ij}^X=(v_i^X)^\dagger v_j^X$ are matrix units, and the classified effects have the ideal matrices as coefficients. The calculation~\eqref{eq:measurement-circuit-units} proves the two intertwining identities. Alice's fixed good flag belongs to the residual space, and~\eqref{eq:canonical-state} gives the state assertion.
\end{proof}

\subsection{The same identities in every commuting representation}
\label{app:rm:commuting}

The compactness argument requires exact identities in a space that contains limits of strategies of unbounded dimension. We therefore extend the preceding lemma to arbitrary commuting representations. The only support property needed is faithfulness in the local support corners; no tensor decomposition or finite-dimensional inverse is required.

Let $\mathfrak A$ be the universal unital $C^*$-algebra generated by projections $a_{xa}$ and $b_{yb}$, for $0\leq x\leq3$, $0\leq y\leq2$, and $0\leq a,b\leq4$, subject to
\[
  \begin{gathered}
    \sum_a a_{xa}=I,\qquad \sum_b b_{yb}=I,
    \qquad [a_{xa},b_{yb}]=0,\\
    a_{xa}a_{xc}=0\ (a\ne c),\qquad
    b_{yb}b_{yd}=0\ (b\ne d).
  \end{gathered}
\]
Thus every pair of commuting projective measurement families gives a representation of $\mathfrak A$. Denote its two local subalgebras by $\mathfrak A_A$ and $\mathfrak A_B$, and define
\begin{equation}
  h_{\mathrm{loss}}=\frac1{12}
    \sum_{\mathsf w(x,y,a,b)=0}a_{xa}b_{yb},
  \qquad 0\leq h_{\mathrm{loss}}\leq I.
  \label{eq:rm:universal-loss}
\end{equation}
The products are positive because the local projections commute. A state on this algebra is a positive linear functional taking the value one at the identity.

\begin{lemma}\label{lem:rm:perfect-commuting}
  Let $\omega$ be any state on $\mathfrak A$ satisfying
  $\omega(h_{\mathrm{loss}})=0$, and let
  $(\pi,\mathcal H,\Omega)$ be its Gelfand--Naimark--Segal representation. Form the joint isometry column $\mathcal V$ by multiplying the commuting coefficients of $\mathcal V_A$ and $\mathcal V_B$. Then
  \begin{align}
    (\Psi_\star\otimes I)\mathcal V\Omega&=\mathcal V\Omega,
      \label{eq:rm:perfect-state}\\
    [\mathcal V a_{xa}-(A^\star_{xa}\otimes I)\mathcal V]\Omega&=0,
      \label{eq:rm:perfect-alice}\\
    [\mathcal V b_{yb}-(I\otimes B^\star_{yb}\otimes I)\mathcal V]\Omega&=0
      \label{eq:rm:perfect-bob}
  \end{align}
  for every label. Here the representation symbol is suppressed, the ideal matrices act on the respective extracted registers, and the final identity includes the residual and flag spaces. No finite-dimensionality assumption is made on $\mathcal H$.
\end{lemma}

\begin{proof}
  We first restrict to local support corners, then reproduce the high- and low-component identifications, and finally evaluate the fixed circuit. Put $\mathcal M_X=\pi(\mathfrak A_X)''$ for $X=A,B$. These von Neumann algebras commute. Let $s_X\in\mathcal M_X$ be the support projection of the normal vector functional
  $T\mapsto\langle\Omega,T\Omega\rangle$ on $\mathcal M_X$. In particular, $s_X\Omega=\Omega$, and the functional is faithful on the corner $s_X\mathcal M_Xs_X$:
  \begin{equation}
    T\in s_X\mathcal M_Xs_X,\quad T\Omega=0
      \quad\Longrightarrow\quad T=0.
    \label{eq:rm:corner-faithfulness}
  \end{equation}
  Indeed, the support is the complement of the largest zero-expectation projection. A positive element of the corner with zero expectation has zero spectral projection on $[t,\infty)$ for every $t>0$, and hence is zero. Apply this to $T^{\dagger}T$. This argument uses normal vector functionals, not a reduced density operator or an assumption that the GNS vector is separating on the original algebra.

  On $s_As_B\mathcal H$, the restrictions of the two corners commute and have common identity $s_As_B$. Their representations there are faithful: an element vanishing on this space annihilates $\Omega$, so~\eqref{eq:rm:corner-faithfulness} applies. Compress every local measurement generator to its corner. The compressed families are POVMs with the same joint probabilities, since cross-party commutation and $s_A\Omega=s_B\Omega=\Omega$ give
  \[
    \langle\Omega,(s_Aa_{xa}s_A)(s_Bb_{yb}s_B)\Omega\rangle
       =\langle\Omega,a_{xa}b_{yb}\Omega\rangle.
  \]
  We next establish the exact classification within these corners.

  The identity~\eqref{eq:exact-agreement} applies to commuting positive contractions on a common Hilbert space. Together with local corner faithfulness, it gives the same agreement and projectivity conclusions as in Appendix~\ref{app:exact-supports}. They produce five projections $P_i,Q_i$ with
  \[
    \sum_{i=0}^4P_i=s_A,\qquad \sum_{i=0}^4Q_i=s_B,
    \qquad P_i\Omega=Q_i\Omega=:\Omega_i.
  \]
  The high-component proof also uses only bounded-operator identities: the cancellations giving~\eqref{eq:high-anticommutation} are valid in the corners, and the resulting commuting reflections identify all four high subspaces. They therefore give matrix units $e_{ij}^X$, $0\leq i,j<4$. Transferring words in reverse order across $\Omega$ is well-defined by~\eqref{eq:rm:corner-faithfulness}, and gives
  \begin{equation}
    e_{ij}^A\Omega=e_{ji}^B\Omega,
    \qquad e_{0i}^Ae_{0i}^B\Omega_i=\Omega_0
       \quad(0\leq i<4).
    \label{eq:rm:high-units}
  \end{equation}
  The diagonal entries are $P_i$ and $Q_i$. These identities replace the tensor-coordinate expression~\eqref{eq:four-components} in the commuting setting. The same joint eigenspaces give Bob's pair supports in~\eqref{eq:bob-pair-supports}.

  To justify the low-component argument, put $p=P_0+P_4$ and $q=Q_0+Q_4$. An Alice pair with high coordinate $r\in\{0,2\}$ and a Bob pair with high line $v\in\{h,k\}$ are mapped into these common pair corners by
  \[
    u_A=e_{0r}^A+P_4,
    \qquad u_B=\sum_{i<4}v_i e_{0i}^B+Q_4.
  \]
  They are commuting partial isometries with final projections $p,q$. The matched-label identities and~\eqref{eq:rm:high-units} give
  $u_Au_B\Omega=v_r\Omega_0+\Omega_4$.
  Thus the four conditional vectors are again
  $z_-,z_-,z_+,z_-$, where
  $z_\pm=\pm\Omega_0/2+\Omega_4$.
  Their local functionals are faithful on the pair corners. For example, if $T\in p\mathcal M_Ap$, the orthogonal opposite-party projections $Q_0,Q_4$ imply
  \begin{equation}
    \norm{Tz_\pm}^{2}
       =\tfrac14\norm{T\Omega_0}^{2}+\norm{T\Omega_4}^{2}.
    \label{eq:rm:pair-faithfulness}
  \end{equation}
  If this vanishes, then $T\Omega=T(\Omega_0+\Omega_4)=0$, so $T=0$ by~\eqref{eq:rm:corner-faithfulness}. The Bob argument is identical.

  Apply~\eqref{eq:exact-agreement} to the four restricted binary POVMs. The differences in their common pair corners are reflections $R,R',C,D'$ satisfying exactly~\eqref{eq:four-binary-relations}, with the present $z_\pm$. Pair faithfulness gives $R=R'$ and $C=D'$. Since
  $z_-=(-P_0+P_4)z_+=(-Q_0+Q_4)z_+$,
  the same calculation gives
  \[
    \{R,-P_0+P_4\}=0,\qquad
    \{C,-Q_0+Q_4\}=0.
  \]
  Consequently $t_A=P_0RP_4$ and $t_B=Q_0CQ_4$ obey
  \[
    t_A^{\dagger}t_A=P_4,\quad t_At_A^{\dagger}=P_0,
    \qquad t_B^{\dagger}t_B=Q_4,\quad t_Bt_B^{\dagger}=Q_0.
  \]
  Taking the high-Alice, low-Bob part of $Rz_-=Cz_-$ and then applying $-t_B$ gives
  \begin{equation}
    t_A(-t_B)\Omega_4=\tfrac12\Omega_0.
    \label{eq:rm:low-unit}
  \end{equation}

  Set $v_i^X=e_{0i}^X$ for $i<4$, $v_4^A=t_A$, and $v_4^B=-t_B$. The operators
  $e_{ij}^X=(v_i^X)^{\dagger}v_j^X$, now for $0\leq i,j\leq4$, form complete matrix-unit systems on the local corners. The pair sums and differences recover every labeled compressed effect as
  \begin{equation}
    a_{xa}^{\mathrm c}=\sum_{i,j=0}^4(A^\star_{xa})_{ij}e_{ij}^A,
    \qquad
    b_{yb}^{\mathrm c}=\sum_{i,j=0}^4(B^\star_{yb})_{ij}e_{ij}^B.
    \label{eq:rm:corner-effects}
  \end{equation}
  Each compressed effect is therefore a projection, so~\eqref{eq:no-leakage} shows that each original generator preserves its local support. Restricting a polynomial in the original generators is therefore the same as evaluating it in the compressed effects.

  The block computation in Lemma~\ref{lem:rm:fixed-circuit} now applies in the local corners. The fixed circuit has coefficients $v_i^X$ there, with Alice's good flag, and
  $v_k^Xe_{ij}^X=\delta_{ki}v_j^X$ proves the intertwining identities on the entire local support. Matched labels,~\eqref{eq:rm:high-units}, and~\eqref{eq:rm:low-unit} give
  \[
    \mathcal V\Omega
      =\left(\sum_{i<4}\ket{ii}+\tfrac12\ket{44}\right)
         \otimes\Omega_0\otimes\ket{\mathrm g}.
  \]
  Its norm is one because $\mathcal V$ is an isometry. In particular,
  $\norm{\Omega_0}^{2}=4/17$, and the extracted factor is $\ket{\psistar}$. This proves all three assertions.
\end{proof}

\subsection{Lifting the measurement and purity assumptions}
\label{app:rm:povm}

Lemma~\ref{lem:rm:perfect-commuting} verifies the zero-loss implication used in Section~\ref{sec:measurement-compactness}. The estimates~\eqref{eq:rm:projective-errors}--\eqref{eq:measurement-common-error} therefore establish a robust self-test for pure projective strategies. We now invoke the assumption-lifting theorem of Baptista et al.~\cite[Theorem 4.1(a)]{BaptistaEtAl}: a robust self-test for pure projective strategies with a pure, full-Schmidt-rank target is a robust self-test for arbitrary mixed states and POVMs. The reference strategy is perfect and has Schmidt rank five on $\C^5\otimes\C^5$, so the theorem applies. Its local-dilation definition~\cite[Definition 2.5]{BaptistaEtAl} controls the state and all single-effect actions simultaneously while leaving the purification environment untouched.

Consequently, there is a dimension-independent modulus $s(\varepsilon)\to0$ for these actions. At zero loss we take $s(0)=0$ by Theorem~\ref{thm:strategy-exact}. Choose the isometries using a minimal purification. Every other purification is obtained by an isometry on the environment, so the same local isometries work for all purifications.

Joint effects require only one further triangle inequality. Write $\xi=\mathcal W\Omega$, $\Phi=\ket{\psistar}\otimes\chi$, and, for fixed labels, set
\[
  X=W_AM_{a|x}W_A^{\dagger},\quad
  Y=W_BN_{b|y}W_B^{\dagger},\quad
  P=A^\star_{xa},\quad Q=B^\star_{yb},
\]
with identities on the other systems suppressed. These are contractions and opposite-party operators commute. Thus
\begin{align}
  \norm{XY\xi-PQ\Phi}
  &\leq\norm{X(Y\xi-Q\Phi)}
      +\norm{Q(X\xi-P\Phi)}+\norm{XQ(\Phi-\xi)}\nonumber\\
  &\leq3s(\varepsilon).
  \label{eq:rm:joint-effects}
\end{align}
Taking
\begin{equation}
  r(\varepsilon)=\min\{2,3s(\varepsilon)\}
  \label{eq:rm:final-modulus}
\end{equation}
proves Theorem~\ref{thm:strategy-robust}. The lifting theorem supplies local isometries for the original POVMs; the fixed-circuit statement above concerns projective strategies.


\begin{thebibliography}{99}
  \raggedright
  \small
  \bibitem{BrassardBroadbentTapp}
  G.~Brassard, A.~Broadbent, and A.~Tapp, \emph{Quantum pseudo-telepathy}, Foundations of Physics \textbf{35}, 1877--1907 (2005). \href{https://doi.org/10.1007/s10701-005-7353-4}{doi:10.1007/s10701-005-7353-4}.

  \bibitem{JungePalazuelos}
  M.~Junge and C.~Palazuelos, \emph{Large violation of Bell inequalities with low entanglement}, Communications in Mathematical Physics \textbf{306}, 695--746 (2011). \href{https://doi.org/10.1007/s00220-011-1296-8}{doi:10.1007/s00220-011-1296-8}.

  \bibitem{VidickWehner}
  T.~Vidick and S.~Wehner, \emph{More nonlocality with less entanglement}, Physical Review A \textbf{83}, 052310 (2011). \href{https://doi.org/10.1103/PhysRevA.83.052310}{doi:10.1103/PhysRevA.83.052310}.

  \bibitem{Mancinska}
  L.~Man\v{c}inska, \emph{Maximally entangled states in pseudo-telepathy games}, \href{https://arxiv.org/abs/1506.07080}{arXiv:1506.07080} (2015).

  \bibitem{RennerEtAl}
  M.~J.~Renner, E.~P.~Lobo, A.~Konderak, R.~Augusiak, and A.~Ac\'in, \emph{All pure entangled states can lead to fully nonlocal correlations}, \href{https://arxiv.org/abs/2604.26605}{arXiv:2604.26605} (2026).

  \bibitem{Lalonde}
  O.~Lalonde, \emph{Maximally entangled states are not complete for pseudo-telepathy}, \href{https://arxiv.org/abs/2608.05378v1}{arXiv:2608.05378v1} (5 August 2026).

  \bibitem{MayersYao}
  D.~Mayers and A.~Yao, \emph{Self testing quantum apparatus}, Quantum Information and Computation \textbf{4}, 273--286 (2004). \href{https://arxiv.org/abs/quant-ph/0307205}{arXiv:quant-ph/0307205}.

  \bibitem{SupicBowles}
  I.~\v{S}upi\'c and J.~Bowles, \emph{Self-testing of quantum systems: a review}, Quantum \textbf{4}, 337 (2020). \href{https://doi.org/10.22331/q-2020-09-30-337}{doi:10.22331/q-2020-09-30-337}.

  \bibitem{McKagueYangScarani}
  M.~McKague, T.~H.~Yang, and V.~Scarani, \emph{Robust self-testing of the singlet}, Journal of Physics A: Mathematical and Theoretical \textbf{45}, 455304 (2012). \href{https://doi.org/10.1088/1751-8113/45/45/455304}{doi:10.1088/1751-8113/45/45/455304}.

  \bibitem{BampsPironio}
  C.~Bamps and S.~Pironio, \emph{Sum-of-squares decompositions for a family of Clauser--Horne--Shimony--Holt-like inequalities and their application to self-testing}, Physical Review A \textbf{91}, 052111 (2015). \href{https://doi.org/10.1103/PhysRevA.91.052111}{doi:10.1103/PhysRevA.91.052111}.

  \bibitem{ColadangeloGohScarani}
  A.~Coladangelo, K.~T.~Goh, and V.~Scarani, \emph{All pure bipartite entangled states can be self-tested}, Nature Communications \textbf{8}, 15485 (2017). \href{https://doi.org/10.1038/ncomms15485}{doi:10.1038/ncomms15485}.

  \bibitem{CuiEtAl}
  D.~Cui, A.~Mehta, H.~Mousavi, and S.~S.~Nezhadi, \emph{A generalization of CHSH and the algebraic structure of optimal strategies}, Quantum \textbf{4}, 346 (2020). \href{https://doi.org/10.22331/q-2020-10-21-346}{doi:10.22331/q-2020-10-21-346}.

  \bibitem{WuEtAl}
  X.~Wu, J.-D.~Bancal, M.~McKague, and V.~Scarani, \emph{Device-independent parallel self-testing of two singlets}, Physical Review A \textbf{93}, 062121 (2016). \href{https://doi.org/10.1103/PhysRevA.93.062121}{doi:10.1103/PhysRevA.93.062121}.

  \bibitem{ColadangeloStark}
  A.~Coladangelo and J.~Stark, \emph{Robust self-testing for linear constraint system games}, \href{https://arxiv.org/abs/1709.09267v2}{arXiv:1709.09267v2} (2019).

  \bibitem{PaddockEtAl}
  C.~Paddock, W.~Slofstra, Y.~Zhao, and Y.~Zhou, \emph{An operator-algebraic formulation of self-testing}, Annales Henri Poincar\'e \textbf{25}, 4283--4319 (2024). \href{https://doi.org/10.1007/s00023-023-01378-y}{doi:10.1007/s00023-023-01378-y}.

  \bibitem{Zhao}
  Y.~Zhao, \emph{Robust self-testing for nonlocal games with robust game algebras}, \href{https://arxiv.org/abs/2411.03259}{arXiv:2411.03259} (2024).

  \bibitem{BaptistaEtAl}
  P.~Baptista, R.~Chen, J.~Kaniewski, D.~R.~Lolck, L.~Man\v{c}inska, T.~G.~Nielsen, and S.~Schmidt, \emph{A mathematical foundation for self-testing: Lifting common assumptions}, Annales Henri Poincar\'e (2025). \href{https://doi.org/10.1007/s00023-025-01642-3}{doi:10.1007/s00023-025-01642-3}. Also available as \href{https://arxiv.org/abs/2310.12662}{arXiv:2310.12662}, whose theorem numbering we use.

  \bibitem{MancinskaSchmidt}
  L.~Man\v{c}inska and S.~Schmidt, \emph{Counterexamples in self-testing}, Quantum \textbf{7}, 1051 (2023). \href{https://doi.org/10.22331/q-2023-07-11-1051}{doi:10.22331/q-2023-07-11-1051}.
\end{thebibliography}
\end{document}